\documentclass[a4paper,UKenglish,cleveref, autoref]{mylipics}
\hideLIPIcs
\nolinenumbers
\usepackage{microtype}%if unwanted, comment out or use option "draft"

\usepackage{amsmath,amssymb}

\usepackage{xspace}
\usepackage{tikz}
\usepackage{graphicx}
\usepackage{color}
\usepackage{xcolor}
\usepackage{mathrsfs}
\title{Approximate Turing Kernelization for Problems Parameterized by Treewidth} %TODO Please add

\author{Eva-Maria C. Hols}{Department of Computer Science, Humboldt-Universit{\"a}t zu Berlin, Germany}{eva-maria.hols@fkie.fraunhofer.de}{https://orcid.org/0000-0002-2832-0722}{}%{Supported by DFG Emmy Noether-grant (KR 4286/1)}
\author{Stefan Kratsch}{Department of Computer Science, Humboldt-Universit{\"a}t zu Berlin, Germany}{kratsch@informatik.hu-berlin.de}{https://orcid.org/0000-0002-0193-7239}{}
\author{Astrid Pieterse}{Department of Computer Science, Humboldt-Universit{\"a}t zu Berlin, Germany}{astrid.pieterse@informatik.hu-berlin.de}{https://orcid.org/0000-0003-3721-6721}{}%{Supported by DFG Emmy Noether-grant (KR 4286/1)}

\funding{\emph{Eva-Maria C. Hols} and \emph{Astrid Pieterse}: Supported by DFG Emmy Noether-grant (KR 4286/1)}

\authorrunning{E.\,C. Hols, S. Kratsch, and A. Pieterse}%TODO mandatory. First: Use abbreviated first/middle names. Second (only in severe cases): Use first author plus 'et al.'

\Copyright{Eva-Maria C. Hols, Stefan Kratsch, Astrid Pieterse}%TODO mandatory, please use full first names. LIPIcs license is "CC-BY";  http://creativecommons.org/licenses/by/3.0/


%TODO mandatory; please add comma-separated list of keywords

\category{}%optional, e.g. invited paper

\relatedversion{}%optional, e.g. full version hosted on arXiv, HAL, or other respository/website
\supplement{}%optional, e.g. related research data, source code, ... hosted on a repository like zenodo, figshare, GitHub, ...

\EventEditors{...}
\EventNoEds{2}
\EventLongTitle{}
\EventShortTitle{}
\EventAcronym{}
\EventYear{}
\EventDate{}
\EventLocation{}
\EventLogo{}
\SeriesVolume{}
\ArticleNo{}
\theoremstyle{plain}
\theoremstyle{definition}

\theoremstyle{remark}

\usepackage[ruled]{algorithm}
\usepackage{algpseudocode}

\newcommand{\prob}[1]{\textsc{#1}}

\newcommand{\friendly}{friendly\xspace}

\newcommand{\Oh}{\mathcal{O}}
\newcommand{\VC}{\prob{Vertex Cover}\xspace}

\newcommand{\classP}{\ensuremath{\mathsf{P}}\xspace}
\newcommand{\NP}{\ensuremath{\mathsf{NP}}\xspace}

\newcommand{\containment}{\ensuremath{\mathsf{NP\subseteq coNP/poly}}\xspace}

\newcommand{\OPTp}{\ensuremath{\mathrm{OPT}_{\P}}\xspace}
\newcommand{\OPTq}{\ensuremath{\mathrm{OPT}_{\Q}}\xspace}
\newcommand{\modOPTp}{\ensuremath{\mathrm{OPT}_{\modP}}\xspace}
\newcommand{\OPTvc}{\ensuremath{\mathrm{OPT}_{\text{VC}}}\xspace}
\newcommand{\OPTcvc}{\ensuremath{\mathrm{OPT}_{\text{CVC}}}\xspace}
\newcommand{\OPTetp}{\ensuremath{\mathrm{OPT}_{\text{ETP}}}\xspace}
\newcommand{\etp}{\ensuremath{\text{ETP}}\xspace}
\newcommand{\modetp}{\ensuremath{\text{ETP}^{\bot}}\xspace}
\newcommand{\modOPTetp}{\ensuremath{\mathrm{OPT}_{\text{ETP}^{\bot}}}\xspace}

\newcommand{\modcvc}{\ensuremath{\text{CVC}^{\bot}}\xspace}
\newcommand{\modOPTcvc}{\ensuremath{\mathrm{OPT}_{\text{CVC}^{\bot}}}\xspace}

\renewcommand{\P}{\ensuremath{\mathcal{P}}\xspace}
\newcommand{\modP}{\ensuremath{\mathcal{P}^{\bot}}\xspace}
\usepackage{enumitem}

\newcommand{\Q}{\ensuremath{\mathcal{Q}}\xspace}

\newcommand{\opt}{\ensuremath{\text{OPT}}}
\newcommand{\optis}{\ensuremath{\text{OPT}_{\text{IS}}}}
\newcommand{\optvc}{\ensuremath{\text{OPT}_{\text{VC}}}}
\newcommand{\optecc}{\ensuremath{\text{OPT}_{\text{ECC}}}}
\newcommand{\todo}[1]{{\color{red}[\textbf{TODO:} #1]}}

\newcommand{\degree}{d}
\newcommand{\defparproblem}[4]{\par
 \vspace{3mm}
\noindent\fbox{
 \begin{minipage}{0.96\textwidth}
 \begin{tabular*}{\textwidth}{@{\extracolsep{\fill}}lr} #1 & {\bf{Parameter:}} #3 \vspace{1mm} \\ \end{tabular*}
 {\textbf{Input:}} #2
	\vspace{1mm}\\%
 {\textbf{Output:}} #4
 \end{minipage}
 }
 \vspace{3mm}
\par
}

\begin{document}

\maketitle

\begin{abstract}
We extend the notion of lossy kernelization, introduced by  Lokshtanov et al. [STOC 2017], to approximate Turing kernelization. An $\alpha$-approximate Turing kernel for a parameterized optimization problem is a polynomial-time algorithm that, when given access to an oracle that outputs $c$-approximate solutions in $\Oh(1)$ time, obtains an $\alpha \cdot c$-approximate solution to the considered problem, using calls to the oracle of size at most $f(k)$ for some function $f$ that only depends on the parameter. 

Using this definition, we show that \textsc{Independent Set} parameterized by treewidth $\ell$ has a $(1+\varepsilon)$-approximate Turing kernel with $\Oh(\frac{\ell^2}{\varepsilon})$ vertices, answering an open question posed by Lokshtanov et al. [STOC 2017]. Furthermore, we give $(1+\varepsilon)$-approximate Turing kernels for the following graph problems parameterized by treewidth: \textsc{Vertex Cover}, \textsc{Edge Clique Cover}, \textsc{Edge-Disjoint Triangle Packing} and \textsc{Connected Vertex Cover}.

We generalize the result for \textsc{Independent Set} and \textsc{Vertex Cover}, by showing that all graph problems that we will call \emph{friendly} admit $(1+\varepsilon)$-approximate Turing kernels of polynomial size when parameterized by treewidth. We use this to obtain approximate Turing kernels for \textsc{Vertex-Disjoint $H$-packing} for connected graphs $H$, \textsc{Clique Cover}, \textsc{Feedback Vertex Set} and \textsc{Edge Dominating Set}.
\end{abstract}

%%%%%%%%%% ABSTRACT in plain text %%%%%%%%%%%%

%% We extend the notion of lossy kernelization, introduced by  Lokshtanov et al. [STOC 2017], to approximate Turing kernelization. An alpha-approximate Turing kernel for a parameterized optimization problem is a polynomial-time algorithm that, when given access to an oracle that outputs c-approximate solutions in O(1) time, obtains an (alpha * c)-approximate solution to the considered problem, using calls to the oracle of size at most f(k) for some function f that only depends on the parameter. 

%% Using this definition, we show that Independent Set parameterized by treewidth l has a (1+epsilon)-approximate Turing kernel with O(l^2/epsilon) vertices, answering an open question posed by Lokshtanov et al. [STOC 2017]. Furthermore, we give (1+epsilon)-approximate Turing kernels for the following graph problems parameterized by treewidth: Vertex Cover, Edge Clique Cover, Edge-Disjoint Triangle Packing and Connected Vertex Cover.

%% We generalize the result for Independent Set and Vertex Cover, by showing that all graph problems that we will call "friendly" admit (1+epsilon)-approximate Turing kernels of polynomial size when parameterized by treewidth. We use this to obtain approximate Turing kernels for Vertex-Disjoint H-packing for connected graphs H, Clique Cover, Feedback Vertex Set and Edge Dominating Set.

%%%%%%%% %%%%%%%%%%%%%%%%%%%%%%%% %%%%%%%%%%%

\newpage
\section{Introduction}

%optional: start with the question that we address and then give the below text to build up to it

Many important computational problems are \NP-hard and, thus, they do not have efficient algorithms unless $\classP=\NP$. At the same time, it is well known that \emph{efficient preprocessing} can greatly speed up (exponential-time) algorithms for solving \NP-hard problems. The notion of a \emph{kernelization} from parameterized complexity has allowed a rigorous and systematic study of this important paradigm. The central idea is to relate the effectiveness of preprocessing to the structure of the input instances, as quantified by suitable parameters. 
% This bypasses the observation that no efficient algorithm can simply all (sufficiently large) inputs of any \NP-hard problem unless $\classP=\NP$.

A \emph{parameterized problem} consists of any (classical) problem together with a choice of one or more parameters; we use $(x,k)$ to denote an instance with input data $x$ and parameter  $k$. A \emph{kernelization} is an efficient algorithm that on input of $(x,k)$ returns an equivalent instance $(x',k')$ of size upper bounded by $f(k)$, where $f$ is a computable function. For a polynomial kernelization we require that the size bound $f(k)$ is polynomially bounded in $k$. The study of which parameterized problems admit (polynomial) kernelizations has turned into a very active research area within parameterized complexity (see, e.g.,
 \cite{DBLP:conf/soda/AgrawalM0Z19, BodlaenderFLPST16_metakernelization,  DBLP:conf/esa/00010SY18, DBLP:conf/wads/ChaplickFGK019, DBLP:conf/focs/FominLMS12, DBLP:conf/stacs/HolsK19,   DBLP:conf/esa/JansenP18,  DBLP:conf/stacs/JansenPL19, DBLP:conf/focs/KratschW12, DBLP:conf/stacs/2019,  DBLP:conf/sofsem/WitteveenBT19}
and the recent book~\cite{fomin_lokshtanov_saurabh_zehavi_2019}). An important catalyst for this development lies in the ability to prove lower bounds for kernelizations, e.g., to conditionally rule out polynomial kernels for a problem, which was initiated through work of Bodlaender et al.~\cite{BodlaenderDFH09} and Fortnow and Santhanam~\cite{FortnowS11}.

%Do not ask me why merging these two paragraphs does not help, but it just doesn't. The mysteries of LaTeX :)
Unfortunately, the lower bound tools have also revealed that many fundamental parameterized problems do not admit polynomial kernelizations (unless \containment and the polynomial hierarchy collapses). These include a variety of problems like \textsc{Connected Vertex Cover}~\cite{DomLS14}, \textsc{Disjoint Cycle Packing}~\cite{DBLP:journals/tcs/BodlaenderTY11}, \textsc{Multicut}~\cite{CyganKPPW14}, and \textsc{$k$-Path}~\cite{BodlaenderDFH09} parameterized by solution size, but also essentially any \NP-hard problem parameterized by width parameters such as treewidth. This has motivated the study of relaxed forms of kernelization, notably Turing kernelization~\cite{Binkele-RaibleFFLSV12} and lossy (or approximate) kernelization~\cite{DBLP:conf/stoc/LokshtanovPRS17}.

Given an input $(x,k)$, a \emph{Turing kernelization} may create $|x|^{\Oh(1)}$ many instances of size at most $f(k)$ each, and the answer for $(x,k)$ may depend on solutions for all those instances. This is best formalized as an efficient algorithm that solves $(x,k)$ while being allowed to ask questions of size at most $f(k)$ to an oracle. A priori, this is much more powerful than regular kernelization, which creates only a single output instance. Nevertheless, there are only few polynomial Turing kernelizations known for problems without (regular) polynomial kernelization (e.g.,~\cite{Binkele-RaibleFFLSV12,DBLP:conf/soda/JansenM15,DBLP:journals/jcss/Jansen17,DBLP:journals/algorithmica/ThomasseTV17}). Moreover, a hardness-based approach of Hermelin et al.~\cite{HermelinKSWW15} gives evidence that many problems are unlikely to admit polynomial Turing kernels.

More recently, Lokshtanov et al.~\cite{DBLP:conf/stoc/LokshtanovPRS17} proposed a framework dedicated to the study of \emph{lossy kernelization}. This relaxes the task of the kernelization by no longer requiring that an optimal solution to the output $(x',k')$ yields an optimal solution for $(x,k)$. Instead, for an $\alpha$-approximate kernelization any $c$-approximate solution to $(x',k')$ can be lifted to an $\alpha\cdot c$-approximate solution for $(x,k)$. Amongst others, they show that \textsc{Connected Vertex Cover} and \textsc{Disjoint Cycle Packing} admit approximate kernelizations. In contrast, they were able to show, e.g., that \textsc{$k$-Path} has no $\alpha$-approximate kernelization for any $\alpha\geq 1$ (unless \containment). Subsequent works have shown approximate kernelizations for other problems~\cite{DBLP:conf/mfcs/EibenHR17,DBLP:journals/siamdm/EibenKMPS19,DBLP:conf/esa/Ramanujan19}, in particular, further problems with connectivity constraints, which are often an obstruction for the existence of polynomial kernelizations.  

Lokshtanov et al.~\cite{DBLP:conf/stoc/LokshtanovPRS17} ask whether \textsc{Independent Set} parameterized by treewidth admits a polynomial-size approximate Turing kernelization with constant approximation ratio. In the present work, we answer this question affirmatively and in fact provide an efficient polynomial size approximate Turing kernelization scheme (EPSATKS). 
Moreover, extending the ideas for \textsc{Independent Set}, we provide similar results for a variety of other problems.

\subparagraph{Our results}%TODO: unifinished but may be going somewhere
We prove that there is an EPSATKS for a wide variety of graph problems when parameterized by treewidth. The simplest problems we consider are the \textsc{Vertex Cover} and \textsc{Independent Set} problem. Observe that both problems parameterized by treewidth can be shown to be $\mathsf{MK}[2]$-hard, by a simple reduction from \textsc{CNF-Sat} with unbounded clause size.\footnote{A variant of the well-known \textsc{NP}-hardness proof of \textsc{Independent Set} (or \textsc{Vertex Cover}) suffices, where we add two vertices $v_x$ and $v_{\bar{x}}$ for each variable $x$ and connect them. Add a clique for each clause, that has a vertex $u_\ell$ for each literal $\ell$ in the clause. Connect $u_\ell$ to $v_x$ if $\ell = \neg x$, connect $u_\ell$ to $v_{\bar{x}}$ if $\ell = x$. Observe that the treewidth is bounded by twice the number of variables.} As such, for both problems we indeed do not expect polynomial Turing kernels~\cite{HermelinKSWW15}. We show that \textsc{Vertex Cover} has a $(1+\varepsilon)$-approximate Turing kernel with $\Oh(\frac{\ell}{\varepsilon})$ vertices, and \textsc{Independent Set} has a kernel with $\Oh(\frac{\ell^2}{\varepsilon})$ vertices. 

Both approximate Turing kernels follow a similar strategy, based on using separators (originating from the tree decomposition) that separate a  piece from the rest of the graph, such that the solution size in this piece is appropriately bounded.  
For this reason, we formulate a set of conditions on a graph problem and we call graph problems that satisfy these conditions \emph{\friendly}.  We then show that all \friendly graph optimization problems have polynomial-size $(1+\varepsilon)$-approximate Turing kernels for all $\varepsilon > 0$, when parameterized by treewidth. Precise bounds on the size of the obtained approximate Turing kernels  depend on properties of the considered problem, such as the smallest-known (approximate) kernel when parameterized by solution size plus treewidth. In particular, applying the general result for \VC indeed shows that it has an EPSATKS of size $\Oh(\frac{\ell}{\varepsilon})$. Using this general technique, we  obtain approximate Turing kernels for \textsc{Clique Cover}, \textsc{Vertex-Disjoint $H$-Packing} for connected graphs $H$, \textsc{Feedback Vertex Set}, and \textsc{Edge Dominating Set}.

Finally, we prove that  \textsc{Edge Clique Cover} and \textsc{Edge-Disjoint Triangle Packing} have an EPSATKS and show that \textsc{Connected Vertex Cover} has a polynomial-size $(1+\varepsilon)$-approximate Turing kernel. These problems do not satisfy our definition of a \friendly problem and require a more problem-specific approach. In particular, for \textsc{Connected Vertex Cover} we will need to consider subconnected tree decompositions~\cite{DBLP:conf/latin/FraigniaudN06} and carefully bound the size difference between locally optimal connected vertex covers, and intersections of (global) connected vertex covers with parts of the graph.

\subparagraph{Overview} 
We start in Section~\ref{sec:specific_problems} by illustrating the general technique using the \textsc{Vertex Cover} problem as an example. We continue by giving the approximate Turing kernels for \textsc{Edge Clique Cover}, \textsc{Connected Vertex Cover}, and \textsc{Edge-Disjoint Triangle Packing}. In Section~\ref{sec:general} we state and prove our general theorem and then show that it allows us to give approximate Turing kernels for a number of different graph problems. 

Finally, in Appendix~\ref{sec:existing-kernels} we show that a number of existing kernels are in fact $1$-approximate kernels, which is needed for some of our proofs. While this is perhaps an expected result, we believe it to be useful for future reference.

\section{Preliminaries}
\label{sec:prelims}  %something 
We use $\mathbb{N}$ to denote the non-negative integers. Let $[n]$ be defined as the set containing the integers $1$ to $n$. 
%\subparagraph{Graphs}
We assume that all graphs are simple and undirected, unless mentioned otherwise. A graph $G$ has vertex set $V(G)$ and edge set $E(G)$. For $v \in V(G)$ we let $d_G(v)$ denote the degree of $v$. For $X\subseteq V(G)$, we use $G[X]$ to denote the graph induced by vertex set $X$, we use $G-X$ to denote $G[V(G)\setminus X]$. For $F \subseteq E(G)$, we use $G \setminus F$ to denote the graph resulting from deleting all edges in $F$ from $G$.

We say that a set $X\subseteq V(G)$ \emph{separates} vertex sets $A\subseteq V(G)$ and $B\subseteq V(G)$ if every path from some vertex in  $A$ to some vertex in $B$ contains a vertex in $X$.

\subparagraph{Treewidth}
We use the standard definition of treewidth:
\begin{definition}[\cite{DBLP:books/sp/CyganFKLMPPS15}]
 A \emph{tree decomposition} of a graph $G$ is a tuple $\mathcal{T} = (T, \{X_t\}_{t \in V(T)})$, where $T$ is a tree in which each node $t \in V(T)$ has an assigned set of vertices $X_t \subseteq V(G)$, also referred to as the \emph{bag} of node $t$, such that the following three conditions hold:
 \begin{itemize}
  \item $\bigcup_{t \in V(T)} X_t = V(G)$, and
  \item for every edge $\{u,v\} \in E(G)$ there exists $t \in V(T)$ such that $u,v \in X_t$, and
  \item for all $v \in V(G)$ the set $T_v := \{t \in V(T) \mid v \in X_t\}$ induces a connected subtree of $T$.
 \end{itemize}
The \emph{width} of a tree decomposition of $G$ is the size of its largest bag minus one. The \emph{treewidth} of $G$ is the minimum width of any tree decomposition of $G$.
\end{definition}

In the remainder of the paper, we will always assume that a tree decomposition~\cite{DBLP:books/sp/CyganFKLMPPS15} is given on input, as treewidth is \NP-hard to compute. If it is not, we may use the result below to obtain an approximation of the treewidth and a corresponding tree decomposition in polynomial time. Doing so will weaken any given size bounds in the paper, as it is not a constant-factor approximation. The theorem below is part of \cite[Theorem 6.4]{DBLP:journals/siamcomp/FeigeHL08}.
\begin{theorem}[{\cite[Theorem 6.4]{DBLP:journals/siamcomp/FeigeHL08}}]
 \label{thm:apprx-tw}
 There exists a polynomial time algorithm that finds a tree decomposition of width at most $\Oh(\sqrt{\log{tw(G)}}\cdot tw(G))$ for a general graph $G$.
\end{theorem}

 Let $\mathcal{T} = (T,\{X_t\}_{t\in V(T)})$ be a tree decomposition. Let $t \in V(T)$, we use $V_t$ to denote the set of vertices from $G$ that are contained in some bag of a node in the subtree of $T$ that is rooted at $t$. %We use $\tw(G)$ to denote the treewidth of $G$. 
 It is well-known that for all $t \in V(T)$, the set $X_t$ separates $V_t$ from the rest of the graph. 
A rooted tree decomposition with root $r$ is said to be \emph{nice} if it satisfies the following properties (cf. \cite{DBLP:books/sp/CyganFKLMPPS15}).
\begin{enumerate}[label=(\roman*)]
 \item $X_r = \emptyset$ and $X_t = \emptyset$ for every leaf $t$ of $T$.
 \item Every other node is of one of the following three types:
 \begin{itemize}
 \item \label{tw:join} The node $t \in V(T)$ has exactly two children $t_1$ and $t_2$, and $X_t = X_{t_1} = X_{t_2}$. We call such a node a \emph{join} node, or
 \item \label{tw:introduce} the node $t \in V(T)$ has exactly one child $t_1$, and there exist $v \in V(G)$ such that $X_t = X_{t_1} \cup \{v\}$ (in this case $t$ is an \emph{introduce} node) or such that $X_{t_1} = X_t \cup \{v\}$ (in which case $t$ is a \emph{forget} node).
 \end{itemize}
\end{enumerate}
One can show that a tree decomposition of a graph $G$ of width $\ell$ can be transformed in polynomial time into
a nice tree decomposition of the same width and with $\Oh(\ell|V (G)|)$ nodes, see for example~\cite{DBLP:books/sp/CyganFKLMPPS15}.

To deal with the \textsc{Connected Vertex Cover} problem we need the tree decomposition to preserve certain connectivity properties. Let a  \emph{subconnected tree decomposition}~\cite{DBLP:conf/latin/FraigniaudN06} be a tree decomposition where $G[V_t]$ is connected for all $t \in V(T)$.  We observe the following.%\footnote{Note that it is not obvious how to transform any given tree decomposition into one that is both subconnected and nice. }

\begin{theorem}[{ cf. \cite[Theorem 1]{DBLP:conf/latin/FraigniaudN06}}]\label{lem:connected-td}
 There is an $\Oh(n\ell^3)$-algorithm that, given a nice tree decomposition on $n$ nodes of width $\ell$ of a connected graph $G$, returns an $\Oh(n\cdot \ell)$-node subconnected tree decomposition of $G$, of width at most $\ell$ such that each node in $T$ has at most $2\ell + 2$ children.
\end{theorem}
\begin{proof}
 Without the additional bound on the degrees of nodes in $T$, the result is immediate from \cite[Theorem 1]{DBLP:conf/latin/FraigniaudN06}. We obtain a subconnected tree decomposition by only executing Phase 1 of Algorithm \verb$make-it-connected$ in \cite{DBLP:conf/latin/FraigniaudN06}. Is is shown in \cite[Claim 1]{DBLP:conf/latin/FraigniaudN06} that  this procedure results in a  tree decomposition of width $\ell$ that is subconnected. It remains to analyze the maximum node degree. The only relevant step of the algorithm is the application of the \verb$split$ operation on nodes $t$ from the original tree. Observe that every node in the original tree is visited at most once, and newly introduced nodes are never \verb$split$. If $t$ has parent $s$, the \verb$split$ operation only modifies the degree of $s$, and any newly introduced nodes. The newly introduced nodes will have degree at most $\degree_T(t)$. In particular, if $s$ had degree $a$ before the \verb$split$ operation on $t$, it will have degree $a - 1 + p$ after the \verb$split$ operation, where $p$ is the number of connected components of $G[V_t]$.
 
  We will show that the number of connected components of $G[V_t]$ is bounded by  $|X_t|$ if $G$ is a connected graph. We do this by showing that each connected component contains at least one vertex from $X_t$. Suppose not. Let $C$ be such a component. But since $C \cap X_t = \emptyset$, and $X_t$ is a separator in $G$, it follows that there are no connections from $C$ to $G[V(G)\setminus V_t]$. If $V_t = V(G)$, then $G[V_t]$ is connected and we are done, otherwise, vertices in $V(G)\setminus V_t$ are not connected to $C$ in $G$, contradicting that $G$ is connected. Thus, $p \leq |X_t| \leq \ell+1$. Since in a nice tree decomposition every node has only two children, in the worst case \verb$split$ is applied to both these children. Thus, every node in $T$ has degree at most $2\ell+2$.
\end{proof}

\subparagraph{Approximation, Kernelization, and Turing Kernelization}
 Before introducing suitable definitions for approximate Turing kernelization, let us
 recall the framework for approximate kernelization by Lokshtanov et al.\ \cite{DBLP:conf/stoc/LokshtanovPRS17} following Fomin et al.\ \cite{fomin_lokshtanov_saurabh_zehavi_2019}.

% We extend the framework for approximate kernelization by Fomin et al. \cite{fomin_lokshtanov_saurabh_zehavi_2019}, we repeat the necessary definitions here.

\begin{definition}[{\cite{fomin_lokshtanov_saurabh_zehavi_2019}}]
 A \emph{parameterized optimization problem} $\Q$ is a computable function 
 \[\Q \colon \Sigma^*\times \mathbb{N}\times \Sigma^* \to \mathbb{R}\cup \{\pm \infty\}.\]
 The \emph{instances} of a parameterized optimization problem are pairs $(I,k)$ where $k$ is the parameter. A \emph{solution} to $(I,k)$ is simply a string $s \in \Sigma^*$, such that $|s| \leq |I|+k$. The \emph{value} of a solution $s$ is given by $\Q(I,k,s)$. 
 Using this, we may define the optimal value for the problem as 
 \[\OPTq(I,k) = \min\{\Q(I,k,s)\mid s \in \Sigma^*, |s|\leq |I| + k\},\]
 for minimization problems and as
 \[\OPTq(I,k) = \max\{\Q(I,k,s)\mid s \in \Sigma^*, |s|\leq |I| + k\},\]
 for maximization problems.
\end{definition}
An \emph{optimization problem} $\P \colon \Sigma^*\times  \Sigma^* \to \mathbb{R}\cup \{\pm \infty\}$ is defined similarly, but without the parameter. In both cases we will say that $s$ is a \emph{solution} for instance $I$, if its value is not $\infty$ (or $-\infty$, in case of maximization problems). 

\begin{definition}
We say that an algorithm for a (regular) minimization problem \P is a \emph{$c$-approximation algorithm} if for all inputs $x$ it returns a solution $s$ such that the value of $s$ is at most $c \cdot \OPTp(x)$. Similarly, for a maximization problem we require that $s$ has value at least $\frac{1}{c} \OPTp(x)$.
\end{definition}

When a problem is parameterized by the value of the optimal solution, the definitions of parameterized optimization problems and lossy kernels will cause problems. %After all, the optimal solution value is exactly what we are asked to compute. 
As such, we use the following interpretation \cite[p.229]{DBLP:conf/stoc/LokshtanovPRS17}. Given an optimization problem $\P$ that we want to parameterize by a sum of (potentially multiple) parameters, one of which is the solution value, we define the following corresponding parameterized optimization problem:
\[\modP(I,k,s) := \min\{\P(I,s), k+1\}.\]
In cases where we consider \P parameterized by the treewidth of the input graph, we simply use $\modP(I,k,s) := \P(I,s)$.

\begin{definition}[$\alpha$-Approximate kernelization {\cite{fomin_lokshtanov_saurabh_zehavi_2019}}]\label{def:lossy-kernel}
 Let $\alpha \geq 1$ be a real number, let $g$ be a computable function and let \Q be a parameterized optimization problem. An \emph{$\alpha$-approximate kernelization} $\mathcal{A}$ of size $g$ for \Q is a pair of polynomial-time algorithms. The first one is called the \emph{reduction algorithm} and computes a map $\mathcal{R}_{\mathcal{A}}\colon \Sigma^* \times \mathbb{N} \to \Sigma^* \times \mathbb{N}$. Given as input an instance $(I,k)$ of \Q, the reduction algorithm computes another instance $(I',k') = \mathcal{R}_{\mathcal{A}}(I,k)$ such that $|I'|,k' \leq g(k)$. 
 
 The second is called the \emph{solution-lifting algorithm}. This algorithm takes as input an instance $(I,k)\in \Sigma^*\times \mathbb{N}$ of \Q, together with $(I',k') := \mathcal{R}_{\mathcal{A}}(I,k)$ and a solution $s'$ to $(I',k')$. In time polynomial in $|I|+|I'|+k+k'+|s|$, it outputs a solution $s$ to $(I,k)$ such that if \Q is a minimization problem, then 
 \[\frac{\Q(I,k,s)}{\OPTq(I,k)} \leq \alpha \cdot \frac{\Q(I',k',s')}{\OPTq(I',k')}.\] 
 For maximization problems we require 
  \[\frac{\Q(I,k,s)}{\OPTq(I,k)} \geq \frac{1}{\alpha} \cdot \frac{\Q(I',k',s')}{\OPTq(I',k')}.\]
\end{definition}

We say that a problem admits a \emph{Polynomial Size Approximate Kernelization Scheme (PSAKS) } \cite{DBLP:conf/stoc/LokshtanovPRS17} if it admits an $\alpha$-approximate polynomial kernel for all $\alpha > 1$. 

We recall the definition of a Turing kernel, so that we can show how to naturally generalize the notion of approximate kernelization to Turing kernels.

\begin{definition}[Turing kernelization {\cite{fomin_lokshtanov_saurabh_zehavi_2019}}]
Let \Q be a parameterized problem and let $f \colon \mathbb{N} \to \mathbb{N}$ be a computable function. A \emph{Turing kernelization} for \Q of size $f$ is an algorithm that decides whether a given instance $(x,k) \in \Sigma^* \times \mathbb{N}$ belongs to \Q in time polynomial in $|x| + k$, when given access to an oracle that decides membership of \Q for any instance $(x',k')$ with $|x'|,k'\leq f(k)$ in a single step.
\end{definition}

In the following definition, we combine the notions of lossy kernelization and Turing kernelization into one, as follows.
\begin{definition}[Approximate Turing kernelization]
Let $\alpha \geq 1$ be a real number, let $f$ be a computable function and let \Q be a parameterized optimization problem. An \emph{$\alpha$-approximate Turing kernel of size $f$} for $\Q$ is an algorithm that, when given access to an oracle that computes a $c$-approximate solution for instances of $\Q$ in a single step, satisfies the following.%FOR JOURNAL VERSION PUT SIZE BOUND HERE
\begin{itemize}
 \item It runs in time polynomial in $|I|+k$, and
 \item given instance $(I,k)$, outputs a solution $s$ such that $\Q(I,k,s) \leq \alpha\cdot c \cdot \OPTq(I,k)$ if $\Q$ is a minimization problem and $\Q(I,k,s) \cdot  \alpha\cdot c \geq \OPTq(I,k)$ is $\Q$ is a minimization problem, and
 \item it only uses oracle-queries of size bounded by $f(k)$.
\end{itemize}
\end{definition}
Note that, in the definition above, the algorithm does not depend on $c$, just like in lossy kernelization. We say that a parameterized optimization problem \Q has an \emph{EPSATKS}  when it has a polynomial-size $(1+\varepsilon)$-approximate Turing kernel for every $\varepsilon > 0$, of size $f(\varepsilon)\cdot \text{poly}(k)$ where $f$ is a function that depends only on $\varepsilon$.

\section{Approximate Turing kernels for specific problems}
\label{sec:specific_problems}
In this section we will give approximate Turing kernels for a number of graph problems parameterized by treewidth. We  start by discussing the \textsc{Vertex Cover} problem, since the approximate Turing kernels for all other problems will follow the same overall structure.

\subsection{Vertex Cover}
\label{subsec:VC}
In this section we discuss an approximate Turing kernel for \textsc{Vertex Cover} parameterized by treewidth $\ell$. The overall idea will be to use the treewidth decomposition of the graph, and find a subtree rooted at a node $t$ such that $G[V_t\setminus X_t]$  has a large (but not too large) vertex cover. A vertex cover of the entire graph will then be obtained by taking a vertex cover of $G[V_t\setminus X_t]$, adding all vertices in $X_t$, and recursing on the graph that remains after removing $V_t$. This produces a correct vertex cover because $X_t$ is a separator in the graph. Furthermore, taking all of $X_t$ into the vertex cover is not problematic as $X_t$ is ensured to be comparatively small. To obtain a vertex cover of $G[V_t\setminus X_t]$, we will use the following lemma.

\begin{lemma}
 \label{lem:approximate-vc}
 Let $G$ be a graph with $\optvc(G)\leq k$. Then there is a polynomial-time algorithm returning  vertex cover of $G$ of size at most $c \cdot \OPTvc(G)$, when given access to $c$-approximate oracle that solves vertex cover on graphs with at most~$\Oh(k)$ vertices. 
\end{lemma}
\begin{proof}
 It is well-known \cite{ChenKJ99VertexCover} that \VC parameterized by solution size $k$ has a kernel with $2k$ vertices, using an LP-based method. We will use this kernelization procedure to prove the lemma. We start by describing the kernel here, cf.\ \cite[Theorem 4]{Fellows2018}. Consider the following LP with variables $x_v$ for all $v \in V(G)$.
 
 \begin{align*}
  &\min \sum_{v \in V(G)} x_v\\
  \intertext{Subject to}
  & x_u + x_v \geq 1 \text{ for all } \{u,v\} \in E(G)\\
  &x_v \in \mathbb{R}, x_v \geq 0 \text{ for all } v \in V(G).
 \end{align*}
    If we would also require $x_v \in \{0,1\}$ for all $v$, this would be equivalent to the \VC problem, we omit this constraint as it would make the LP hard to solve. Obtain an optimal solution for the relaxed problem. 
 Define $V_{0} := \{v \in V(G) \mid x_v <\frac{1}{2}\}$, $V_{\frac{1}{2}}:= \{v \in V(G) \mid x_v = \frac{1}{2}\}$, 
 and $V_1 := \{v \in V(G) \mid x_v > \frac{1}{2}\}$. It has been shown \cite{Nemhauser1974} that there always exists a minimum vertex cover $S$ in $G$ such that $V_1 \subseteq S \subseteq V_1 \cup V_{\frac{1}{2}}$. 
 Let $G' := G[V_{\frac{1}{2}}]$, and observe \cite{Fellows2018} that $|V(G')| \leq 2k$. 
 Apply the $c$-approximate oracle to obtain a vertex cover $S'$ in $G'$. 
 Let $S:= S'\cup V_{1}$. First of all, we show that $S$ is a vertex cover in $G$. 
 Any edge with at least one endpoint in $V_1$ is covered by definition, and any edge within $V_{\frac{1}{2}}$ is covered since $S'$ is a vertex cover of $G'$. 
 This leaves edges within $V_0$ and between $V_{\frac{1}{2}}$ and $V_{0}$, but there are no such edges as these would imply that the chosen solution to the LP is not correct. So $S$ is a vertex cover of $G$. Using that there exists an optimal solution $S^*$ such that  
 $ V_1 \subseteq S^* \subseteq V_1 \cup V_{\frac{1}{2}}$, we observe \begin{align*}
|S| = |V_1| + |S'| &\leq |V_1| + c \cdot \optvc(G')\\
&\leq|S^* \cap V_1| + c \cdot |S^* \cap V_{\frac{1}{2}}| \leq c \cdot |S^*| = c \cdot \optvc(G).\qedhere
 \end{align*}
\end{proof}

Using this, we can now give the $(1+\varepsilon)$-approximate Turing kernel for \VC. While the theorem statement requires $\varepsilon \leq 1$, this does not really impose a restriction: if $\varepsilon > 1$ one may simply reset it to be $1$. It simply shows that the bounds do not continue improving indefinitely as $\varepsilon$ grows larger than $1$. Note however that \textsc{Vertex Cover} is $2$-approximable in polynomial time, such that choosing $\varepsilon$ larger than one is likely not useful. 

\begin{theorem}\label{thm:vc-tk}
For every $0 < \varepsilon\leq 1$, \textsc{Vertex Cover} parameterized by treewidth $\ell$ has a $(1+\varepsilon)$-approximate Turing kernel  with $\Oh(\frac{\ell}{\varepsilon})$ vertices. 
\end{theorem}
\begin{proof}
Consider Algorithm~\ref{alg:approxVC}, we use the well-known $2$-approximation algorithm for \VC. First of all, we show how to do Step~\ref{step:VC:find_t} of the algorithm efficiently.

 \begin{algorithm}
\caption{An approximate Turing kernel for \textsc{Vertex Cover}.}
\label{alg:approxVC}
\begin{algorithmic}[1]
\Procedure{$\textsc{ApproxVC}(G,\mathcal{T},\varepsilon)$}{}
\State Turn $\mathcal{T}$ into a nice tree decomposition of $G$
\State Obtain a $2$-approximate solution $\tilde{S}$ for \textsc{VC} in $G$
\If{$|\tilde{S}| \leq \frac{8(\ell + 1)}{\varepsilon}$}
    \State Determine a $c$-approximate solution $S$ to VC in $G$ using Lemma~\ref{lem:approximate-vc} \label{step:VC:easy}
    \State \Return $S$
\Else
    \State \label{step:VC:find_t}Find $t \in V(T)$ s.t.  $\frac{(\ell + 1)}{\varepsilon} \leq \optvc(G[V_t\setminus X_t]) \leq \frac{8(\ell+1)}{\varepsilon}$
    \State Determine a $c$-approximate solution $S_t$ to VC in $G[V_t\setminus X_t]$ using Lemma~\ref{lem:approximate-vc}
    \State $G' \gets G - V_t$
    \State Let $\mathcal{T}'$ be $\mathcal{T}$ after removing the subtree rooted at $t$ and all vertices in $X_t$
    \State $S' \gets \textsc{ApproxVC}(G',\mathcal{T}',\varepsilon)$
    \State \Return $S_t \cup X_t \cup S'$\label{step:VC:return2}
\EndIf
\EndProcedure
\end{algorithmic}
\end{algorithm}

\begin{claim} \label{claim:vc:find-t}
There is a polynomial-time algorithm that, given graph $G$ such that  $\optvc(G) \geq \frac{(\ell+1)}{\varepsilon}$, with a nice tree decomposition $\mathcal{T}$ of width at most $\ell$, outputs 
 a node $t \in V(T)$ such that $\frac{(\ell+1)}{\varepsilon}\leq \optvc(G[V_t\setminus X_t])\leq \frac{8(\ell+1)}{\varepsilon}$.
\end{claim}
\begin{claimproof}
 Let $\mathcal{T}$ be a nice tree decomposition with root $r$. We start from  $t:= r$, maintaining that $\optvc(G[V_t\setminus X_t]) \geq \frac{\ell+1}{\varepsilon}$. Note that this is initially true since $G_r = G$.

Decide if the $2$-approximation returns a vertex cover of size at most $\frac{8(\ell + 1)}{\varepsilon}$ for $G[V_t\setminus X_t]$. If yes, we are done. If not, then $\optvc(G[V_t\setminus X_t]) > \frac{4(\ell + 1)}{\varepsilon}$. We show that $t$ has a child on which we will recurse. We do a case distinction on the type of node of $t$.
\begin{itemize}
 \item $t$ is a leaf node. In this case, $|V_t\setminus X_t|=0$, contradicting that $\optvc(G[V_t\setminus X_t]) > \frac{4(\ell + 1)}{\varepsilon} \geq 0$.
 \item $t$ is a forget or introduce node. This implies $t$ has one child $t_1$ and the size of $V_{t}\setminus X_t$ and $V_{t_1} \setminus X_{t_1}$ differs by at most one. Therefore, $\optvc(G[V_{t_1}\setminus X_{t_1}]) \geq \optvc(G[V_t\setminus X_{t}])-1 \geq\optvc(G[V_t\setminus X_{t}])/2 $.
 \item $t$ is a join node with children $t_1$ and $t_2$. Observe that $G[V_t \setminus X_t]$ is the disjoint union of the graphs $G[V_{t_1}\setminus X_{t_1}]$ and $G[V_{t_2} \setminus X_{t_2}]$ (note $X_t = X_{t_1} = X_{t_2}$). As such, for one of the two children, without loss of generality let this be $t_1$, running the $2$-approximation algorithm for vertex cover returns a value of at least $\optvc(G[V_t\setminus X_{t}])/2$, meaning that $\optvc(G[V_{t_1}\setminus X_{t_1}])\geq \optvc(G[V_t\setminus X_{t}])/4$.
\end{itemize}
Thus, there is a child $t_1$ such that $\optvc(G[V_{t_1}\setminus X_{t_1}]) \geq \optvc(G[V_t\setminus X_{t}])/4 \geq \frac{\ell + 1}{\varepsilon}$. Continue with $t := t_1$.
\end{claimproof}

We will now show the correctness of the algorithm by induction on $|V(G)|$. Let $G$ be a graph with nice tree decomposition $\mathcal{T}$. If the algorithm returns a \textsc{Vertex Cover} in Step~\ref{step:VC:easy}, the result is immediate. If not, then it follows that the algorithm returns in Step~\ref{step:VC:return2}, and that $\optvc(G) > \frac{4(\ell+1)}{\varepsilon}$. The algorithm then returns a vertex cover $S_t$ for $G[V_t\setminus X_t]$ together with $X_t$ and a vertex cover $S'=\textsc{ApproxVC}(G',\mathcal{T}',\varepsilon)$ in the remainder of the graph. It is easy to see that the returned set is indeed a vertex cover of the graph. Furthermore, one may verify that the oracle is only used for graphs with at most $\Oh(\frac{\ell}{\varepsilon})$ vertices. It remains to verify the approximation ratio. Recall that $G' := G- V_t$. Then
\begin{align*}
 |S_t| + |S'| + |X_t| &\leq c \cdot \optvc(G[V_t \setminus X_t]) + c\cdot (1+\varepsilon)\cdot\optvc(G') + \ell + 1\\
&\leq  c \cdot (1+\varepsilon) \cdot\optvc(G[V_t \setminus X_t]) + c \cdot (1+\varepsilon)\cdot\optvc(G')\\
&\leq c \cdot (1+\varepsilon )\cdot \optvc(G).\qedhere
\end{align*}
\end{proof}

\subsection{Edge Clique Cover}\label{app:ecc}
In this section, we obtain an approximate Turing kernel for \textsc{Edge Clique Cover}, which is defined as follows.

\defparproblem{\textsc{Edge Clique Cover} (ECC)}{A graph $G$ with tree decomposition $\mathcal{T}$ of width $\ell$.}{$\ell$}{The minimum value for $k \in \mathbb{N}$ such that there exists a family $S$ of subsets of $V(G)$ such that $|S| \leq k$,  $G[C]$ is a clique for all $C \in S$, and  for all $\{u,v\} \in E(G)$ there exists $C \in S$ such that $u,v\in S$?}

To obtain an approximate Turing kernel, we will separate suitably-sized subtrees from the graph using the tree decomposition, as we did in the approximate Turing kernel for \textsc{Vertex Cover}. To show that this results in the desired approximation bound, we will need the following lemma. It basically shows that if we find a node $t$ of the tree decomposition such that $X_t$ is ``small'' compared to $\opt(V_t)$, we will be able to combine an edge clique cover in $G[V_t]$ with one in $G - (V_t \setminus X_t)$ to obtain a clique cover of the entire graph that is not too far from optimal.

\begin{lemma}\label{lem:ecc:sumleqopt}
 Let $G$ be a graph, let $X_1,X_2 \subseteq V(G)$ such that $X_1 \cup X_2 = V(G)$ and $X = X_1\cap X_2$ separates $X_1$ from $X_2$ in $G$. Then \[\optecc(G) \geq \optecc(G[X_1]) + \optecc(G[X_2]) - \binom{|X|}{2}.\]
\end{lemma}
\begin{proof}
 Let $S$ be an edge clique cover of $G$. We show how to obtain clique covers $S_1$ and $S_2$ for $G[X_1]$ and $G[X_2]$ such that $|S_1| + |S_2| \leq |S| + \binom{|X|}{2}$. First define 
 \[S_1' := \{C  \mid C \cap (X_1\setminus X) \neq \emptyset, C \in S\} \cup \{C \mid C \subseteq X, C \in S\},\]%technically there is no need for this last part
similarly, define 
 \[S_2' := \{C  \mid C \cap (X_2\setminus X) \neq \emptyset, C \in S\}.\]
 For $j\in[2]$, define $S_j := S_j' \cup S_j''$, where $S_j'' := \{\{u,v\} \in E(G[X]) \mid \{u,v\} \text{ not covered by } S_j'\}$. 
 
 We start by showing that $S_j$ is an edge clique cover of $G[X_j]$ for $i \in [2]$. First of all, we will verify that  $C \subseteq X_j$ and that $C$ forms a clique in $G[X_j]$ for all $C \in S_j$. For $C \in S_j''$ this is trivial, for $C \in S_j'$,  observe that $C$ is a clique in $G$ and any clique in $G$ containing a vertex from $X_j\setminus X$ cannot contain a vertex from $V(G) \setminus X_j$, since $X$ is a separator.  Thus $C \subseteq X_j$. The fact that $C$ is a clique in $G[X_j]$ is immediate from $C$ being a clique in $G$.
 
 It remains to show that $S_j$ covers all edges in $G[X_j]$. Let $\{u,v\} \in E(G[X_j])$. If $u,v \in X$, then the edge is covered by definition. Without loss of generality, suppose $u \in X_j \setminus X$. Let $C \in S$ be a clique that covered edge $\{u,v\}$. Then clearly $u \in C \cap (X_j \setminus X)$ and thus $C \cap (X_j\setminus X) \neq \emptyset$, implying $C  \in S_j$. Thus, the edge $\{u,v\}$ is indeed covered by $S_j$.
 
 It remains to show that $|S_1| + |S_2| \leq |S| + \binom{|X|}{2}$. Start by observing that $|S_1'| + |S_2'| \leq |S|$, since a clique cannot contain both a vertex from $X_1\setminus X$ and $X_2 \setminus X$. Since every edge $\{u,v\} \in E(G[X])$ is covered by $S$, it is easy to observe from the definition that $\{u,v\}$ is covered by $S_1'$ or $S_2'$. As such, $S_1''\cap S_2'' = \emptyset$. Since $G[X]$ has at most $\binom{|X|}{2}$ edges, it follows that $|S_1''| + |S_2''| \leq \binom{|X|}{2}$ and indeed $|S_1| + |S_2| \leq |S_1'| + |S_1''| + |S_2'| + |S_2''| \leq |S| + \binom{|X|}{2}$. 
 \end{proof}

Before giving the approximate Turing kernel, we show that there exists a node $t$ in the tree decomposition such that the size of the subtree rooted at $t$ falls within certain size bounds. We use this to split off subtrees, similar to the strategy we used for \textsc{Vertex Cover} earlier.
 
 \begin{lemma}
  \label{lem:chop-of-part-ecc} There is a polynomial-time algorithm that, given a graph $G$ with $|V(G)| \geq 2\frac{1+\varepsilon}{\varepsilon} (\ell + 1)^4 $, a nice tree decomposition $\mathcal{T}$ of width $\ell$, and $\varepsilon > 0$, outputs a node $t \in V(T)$ such that $2\frac{1+\varepsilon}{\varepsilon} (\ell + 1)^4  \leq |V_t \setminus X_t| \leq 4\frac{1+\varepsilon}{\varepsilon} (\ell + 1)^4 $. 
 \end{lemma}
 \begin{proof}
If $|V(G)| \leq 4\frac{1+\varepsilon}{\varepsilon} (\ell + 1)^4 $, we simply output the root $r$ of the tree decomposition, observe $X_r = \emptyset$ and thus $V_r\setminus X_r = V(G)$. Otherwise, we search  through the tree decomposition to find the right node, as follows. Start from $t := r$ and suppose we are currently at node $t$, such that  $|V_t \setminus X_t| \geq 2\frac{1+\varepsilon}{\varepsilon} (\ell + 1)^4 $. If $|V_t\setminus X_t| \leq 4\frac{1+\varepsilon}{\varepsilon} (\ell + 1)^4 $, we are done.  Otherwise, we show that one of the children $t'$ of $t$ has the property that $|V_{t'} \setminus X_{t'}| \geq 2\frac{1+\varepsilon}{\varepsilon} (\ell + 1)^4 $. Observe that since $\mathcal{T}$ is nice, $t$ has at most two children. If $t$ has exactly one child $t'$, the difference between $|V_t \setminus X_t|$ and $|V_{t'} \setminus X_{t'}|$ is at most one, such that indeed $|V_{t'} \setminus X_{t'}| \geq 4\frac{1+\varepsilon}{\varepsilon} (\ell + 1)^4-1 \geq  2\frac{1+\varepsilon}{\varepsilon} (\ell + 1)^4$. Otherwise, $t$ is a join node and $V_t\setminus X_t = (V_{t_1}\setminus X_{t_1}) \cup (V_{t_2}\setminus X_{t_2})$ for the children $t_1$ and $t_2$ of $t$. Suppose without loss of generality that $(V_{t_1}\setminus X_{t_1}) \geq (V_{t_2}\setminus X_{t_2})$, then $|V_{t_1}\setminus X_{t_1}| \geq |V_t\setminus X_t|/2 \geq 
2\frac{1+\varepsilon}{\varepsilon} (\ell + 1)^4$.
\end{proof}
Using the lemma above, we can now give the approximate Turing kernel for \textsc{Edge Clique Cover}. 

\begin{theorem}\label{thm:ecc}
For every $0 < \varepsilon\leq 1$, \textsc{Edge Clique Cover} parameterized by treewidth $\ell$ has a $(1+\varepsilon)$-approximate Turing kernel with $\Oh(\frac{\ell^4}{\varepsilon})$ vertices.
\end{theorem}
\begin{proof}
 Consider Algorithm~\ref{alg:approxECC}, we show that it is a $(1+\varepsilon)$-approximate Turing kernel for ECC. Observe that Step~\ref{step:ECC:split} can be done efficiently while maintaining a valid tree decomposition, as one may simply restrict the bags of the decomposition to the relevant connected component of $G$.
  It is easy to verify that the procedure runs in polynomial time, using that $|V_t\setminus X_t|$ is always non-empty and thus the recursive call is on a strictly smaller graph. Finally, we can verify the size-bound, as the oracle is only applied to $G$ if $|V(G)| \leq \Oh(\frac{\ell^4}{\varepsilon})$ or to $G[V_t]$ when $|V_t\setminus X_t| \leq \Oh(\frac{\ell^4}{\varepsilon})$, implying that $|V_t| \leq |V_t\setminus X_t| + \ell + 1 = \Oh(\frac{\ell^4}{\varepsilon})$.

\begin{algorithm}
\caption{An approximate Turing Kernel for \textsc{Edge Clique Cover}.}
\label{alg:approxECC}
\begin{algorithmic}[1]
\Procedure{$\textsc{ApproxECC}(G,\mathcal{T},\varepsilon)$}{}
\State   \label{step:ECC:split}If $G$ is not connected, split $G$ into its connected components and treat them separately.
\State Turn $\mathcal{T}$ into a nice tree decomposition.
\If{$|V(G)| \leq \frac{2 (1+\varepsilon)}{\varepsilon}(\ell+1)^4$}
    \State Apply the $c$-approximate oracle to obtain an ECC $S$ of $G$
    \State \Return $S$ \label{step:ecc:return1}
\Else
    \State \label{step:ECC:find_t}Find $t \in V(T)$ s.t.  $2\frac{(1+\varepsilon)}{\varepsilon}(\ell+1)^4 \leq |V_t\setminus X_t| \leq \frac{4 (1+\varepsilon)}{\varepsilon}(\ell+1)^4$ (by Lemma~\ref{lem:chop-of-part-ecc})
    \State Determine a $c$-approximate solution $S_t$ to ECC in $G[V_t]$  using the oracle
    \State $G' \gets G - (V_t\setminus X_t)$
    \State Let $\mathcal{T}'$ be $\mathcal{T}$ after removing the subtree rooted at $t$ except for $t$
    \State $S' \gets \textsc{ApproxECC}(G',\mathcal{T}',\varepsilon)$
    \State \Return $S_t \cup S'$
\EndIf
\EndProcedure
\end{algorithmic}
\end{algorithm}

We continue by showing that Algorithm~\ref{alg:approxECC} returns an edge clique cover of $G$. If the algorithm returns in Step~\ref{step:ecc:return1}, this is immediate. Otherwise, observe that since $X_t$ separates $V_t$ and $V(G')$ in $G$, it follows that any edge in $G$ is in $E(G[V_t])$ or in $E(G')$. Thus, such an edge is covered by $S_t$ or $S'$, implying that $S = S_t \cup S'$ is an edge clique cover of $G$. We now bound $|S_t| + |S'|$, to show that the algorithm indeed approximates the optimum ECC.  

\begin{align*}
 |S_t| + |S'| &\leq c \cdot \optecc(G[V_t]) + |S'|\\
 &= c\cdot ( 1 + \varepsilon)\cdot \optecc(G[V_t]) - c \cdot \varepsilon \cdot \optecc(G[V_t]) + |S'|\\
 \intertext{Observe that every clique covers at most $\binom{\ell+1}{2}$ edges, since it has at most $\ell+1$ vertices, since the treewidth of $G$ is bounded by $\ell$. Thus $\optecc(G[V_t]) \geq |E(G[V_t])|/\binom{\ell+1}{2}$.}
  &\leq c\cdot ( 1 + \varepsilon)\cdot \optecc(G[V_t]) - c \cdot \varepsilon \cdot |E(G[V_t])|/\binom{\ell+1}{2} + |S'|\\
  \intertext{Observe that $V_t\setminus X_t$ cannot contain vertices that are isolated in $G[V_t]$, since $G$ is connected and $X_t$ separates $V_t$ from the remainder of $G$. Thus, $|E(G[V_t])| \geq |V_t\setminus X_t|/2$.}
    &\leq c\cdot ( 1 + \varepsilon)\cdot \optecc(G[V_t]) - c \cdot \varepsilon \cdot \frac{|V_t\setminus X_t|}{2(\ell+1)^2} + |S'|\\
    &\leq c\cdot (1 + \varepsilon)\cdot  \optecc(G[V_t]) - c \cdot (1+\varepsilon) \cdot (\ell + 1)^2 + |S'|\\
    \intertext{using $\ell+1 \geq |X_t|$}
 &\leq c\cdot (1 + \varepsilon)\cdot  \optecc(G[V_t]) - c \cdot (1+\varepsilon) \cdot \binom{|X_t|}{2} + |S'|\\
 &\leq c\cdot (1 + \varepsilon) \cdot (\optecc(G[V_t]) + \optecc(G') - \binom{|X_t|}{2})\\
 \intertext{By Lemma~\ref{lem:ecc:sumleqopt}}
 &\leq c\cdot (1+\varepsilon)\cdot\optecc(G).\qedhere
\end{align*}
\end{proof}

\subsection{Edge-Disjoint Triangle Packing}\label{app:etp}
In this section we give an approximate Turing kernel for the \textsc{Edge-Disjoint Triangle Packing} problem, defined as follows.
\defparproblem{\textsc{Edge-Disjoint Triangle Packing} (ETP)}{A graph $G$ with tree decomposition $\mathcal{T}$ of width $\ell$.}{$\ell$}{The maximum value for $k \in \mathbb{N}$ such that there exists a family $S$ of size-$3$ subsets of $V(G)$ such that $|S| \geq k$,  $G[X]$ is a triangle for all $X \in S$, and $X$ and $Y$ are edge-disjoint for all $X,Y \in S$?}
Observe that the problem has a $3$-approximation by taking any maximal edge-disjoint triangle packing $S$, which can be greedily constructed. This packing then uses  $3|S|$ edges. If there is a solution $S'$ with $|S'| > 3|S|$, then there is a triangle in $S'$ that contains no edge covered by $S$, contradicting that $S$ is maximal. We now give the approximate Turing kernel.

\begin{theorem}\label{thm:etp}
 For every $0 \leq \varepsilon \leq 1$,  \textsc{Edge-Disjoint Triangle Packing}  parameterized by treewidth $\ell$, has a $(1+\varepsilon)$-approximate Turing kernel with $\Oh(\frac{\ell^2}{\varepsilon})$ vertices.
\end{theorem}
\begin{proof}
We start by proving the following claim. 

\begin{claim}\label{claim:etp}
Let $G$ be a graph with $\OPTetp(G) \leq k$. There is a polynomial-time algorithm that when given access to a $c$-approximate oracle, outputs a $c$-approximate solution for $G$ using calls to the oracle with at most $\Oh(k)$ vertices.
\end{claim}
\begin{claimproof}
Start by computing a $3$-approximate solution $\tilde{S}$ to \textsc{ETP} in $G$. Note that $|\tilde{S}| \geq \frac{1}{3} \OPTetp(G)$. Obtain graph $(G',k')$ by applying the $1$-approximate kernel from Lemma~\ref{lem:kernel:EDTP} to $(G,3|\tilde{S}|)$. Apply the $c$-approximate oracle to obtain a solution $S'$ in $G'$. Apply the solution lifting algorithm to obtain solution $S$ in $G$. Let $\hat{S}$ be the largest of $S$ and $\tilde{S}$, output $\hat{S}$. It remains to verify that $\hat{S}$ is a $c$-approximate solution. Note that $|\hat{S}| \geq \modetp(G,3|\tilde{S}|,S)$ and $\OPTetp(G) \leq 3|\tilde{S}|$, such that 
\begin{align*}
\frac{|\hat{S}|}{\OPTetp(G)} &\geq \frac{\modetp(G,3|\tilde{S}|,S)}{\modOPTetp(G,3|\tilde{S}|)} \geq \frac{\modetp(G',k',S')}{\modOPTetp(G',k')}. 
\end{align*}
We consider two options. If $|S'| > k'$, then immediately $\OPTetp(G') > k'$ and thus \[\frac{\modetp(G',k',S')}{\modOPTetp(G',k')} = 1 \geq \frac{1}{c}.\] Otherwise, we get that $\modetp(G',k',S')=\etp(G',S')$, and \[\frac{\modetp(G',k',S')}{\modOPTetp(G',k')} \geq \frac{|S'|}{\OPTetp(G')} \geq \frac{1}{c}.\qedhere\]
\end{claimproof}
We now describe the algorithm.  Start by computing a $3$-approximate solution $\tilde{S}$ to \textsc{Edge-Disjoint Triangle Packing} in~$G$. If $|\tilde{S}| \leq 18\frac{(\ell+1)^2}{\varepsilon}$, we obtain an approximate solution to triangle packing using Claim~\ref{claim:etp}. 

Otherwise, for $t \in V(T)$ define $G_t$ as $G[V_t]\setminus E(G[X_t])$, i.e., the graph $G[V_t]$ from which the edges between vertices in $X_t$ have been removed. We show how to find $t \in T$ such that \[\frac{(\ell+1)^2}{\varepsilon} \leq \OPTetp(G_t) \leq 18\frac{(\ell+1)^2}{\varepsilon},\] together with an approximate solution $S_t$ in $G_t$. Start with $t := r$, observe that initially  $\OPTetp(G_t) > \frac{18(\ell+1)^2}{\varepsilon}$ since $G_r = G$ and $\OPTetp(G_t) \geq |\tilde{S}|$. So suppose we are at some node $t$ with $\OPTetp(G_t) \geq \frac{(\ell+1)^2}{\varepsilon}$. Compute a $3$-approximate solution in $G_t$. If this solution has value at most $\frac{6(\ell+1)^2}{\varepsilon}$,  we obtain an approximate solution $S_t$ to triangle packing in $G_t$ using Claim~\ref{claim:etp}. Otherwise, we will recurse on a child $t_1$ of $t$ for which $\OPTetp(G_{t_1}) \geq \frac{(\ell+1)^2}{\varepsilon}$, we show how to find such a child  by doing a case distinction on the type of node of $t$.
\begin{itemize}
\item $t$ is a leaf node. This is a contradiction with the assumption that $\OPTetp(G_t) > 6\frac{(\ell+1)^2}{\varepsilon}$, since $G_t$ is empty.
\item $t$ has exactly one child $t_1$ and $X_t = X_{t_1} \cup \{v\}$ for some $v\in V(G)$. This means in particular that $G_{t_1} = G_t - \{v\}$. Furthermore, we can show that $v$ is isolated in $G_t$. After all, there are no edges between vertices in $X_t$ and $v$ by definition of $G_t$. Furthermore, there are no edges between $v$ and vertices not in $X_t$, by correctness of the tree decomposition. Therefore, trivially, $\OPTetp(G_t) = \OPTetp(G_{t_1})$ and we continue with $t \gets t_1$.
\item $t$ has exactly one child $t_1$ and $X_t = X_{t_1} \setminus \{v\}$ for some $v\in V(G)$. In this case $G_{t_1}$ can be obtained by $G_t$ by removing all edges between vertices in $v$ and vertices in $X_t$. This removes at most $(\ell+1)$ edges from the graph, and thus $\OPTetp(G_{t_1}) \geq \OPTetp(G_t) - \ell \geq \frac{(\ell+1)^2}{\varepsilon}$, and we continue with $t \gets t_1$.
\item $t$ is a join node with children $t_1$ and $t_2$. Observe that $X_t$ separates $G_t$ and that $\OPTetp(G_t) = \OPTetp(G_{t_1}) + \OPTetp(G_{t_2})$. As such, there is a child of $G_t$, w.l.o.g. let this be $t_1$, such that $\OPTetp(G_{t_1}) \geq \OPTetp(G_t)/2 \geq \frac{3(\ell+1)^2}{\varepsilon}$. Using the $3$-approximation on both children,  find a child where the returned solution size is at least $\frac{3(\ell+1)^2}{3\varepsilon}=\frac{(\ell+1)^2}{\varepsilon}$. Continue with this child.
\end{itemize}
Using $t$ and the obtained solution $S_t$ in $G_t$, we now do the following. Let $G' := G - (V_t\setminus X_t)$. Obtain a solution $S'$ in $G'$ using the algorithm above on the smaller graph $G'$. Output $S := S_t\cup S'$. Since $G'$ and $G_t$ are edge-disjoint subgraphs of $G$, it is easy to observe that $S$ is an edge-disjoint triangle packing in $G$.

It remains to show that $S$ has the desired size. Observe that the size of an edge-disjoint triangle packing in $G$ can be bounded by considering the triangles whose edges are in $G_t$, those whose edges are in $G'$, and those with at least one edge with both endpoints in $X_t$. Using that there are at most $\binom{X_t}{2}$ edges between vertices in $X_t$, we get
\begin{align*}
 \OPTetp(G) &\leq \OPTetp(G_t) + \OPTetp(G') + \binom{X_t}{2}\\
 &\leq (1+\varepsilon)\OPTetp(G_t) + \OPTetp(G')\\
 &\leq c \cdot (1+\varepsilon) |S_t| + c \cdot (1+\varepsilon) |S'| \\
 &\leq c \cdot (1+\varepsilon) |S|. \qedhere
\end{align*}
\end{proof}

The strategy used  to obtain a kernel for \textsc{Edge-Disjoint Triangle Packing} can be generalized to packing larger cliques, as long as these problems have polynomial kernels when parameterized by solution size. Generalizing to the more general question of packing edge-disjoint copies of some other graph $H$ may be more difficult. In this case, there can be copies of $H$ that have vertices in both sides of the graph after removing the edges within a separator, and one needs to be careful to not discard too many of these.

\subsection{Connected Vertex Cover}

The \textsc{Connected Vertex Cover} (CVC) problem asks, given a graph $G$ and tree decomposition $\mathcal{T}$, for the minimum size of a vertex cover $S$ in $G$ such that $G[S]$ is connected. It is known that CVC has a $(1+\varepsilon)$-approximate kernel of polynomial size~\cite{DBLP:conf/stoc/LokshtanovPRS17}.
\begin{theorem}[{\cite{DBLP:conf/stoc/LokshtanovPRS17}}]\label{thm:cvc-psaks}
\textsc{Connected Vertex Cover} parameterized by solution size $k$ admits a strict time efficient PSAKS with $\Oh(k^{\lceil\frac{\alpha}{\alpha-1}\rceil}+k^2)$ vertices.
\end{theorem}

To obtain an approximate Turing kernel, we will use a similar strategy to the Turing kernel for \textsc{Vertex Cover} described in Theorem~\ref{thm:vc-tk}. However, the connectivity constraint makes this kernel somewhat more complicated. We deal with this by changing the procedure in two places. First of all, we will use a subconnected tree decomposition, to ensure that $G[V_t]$ is connected for any node $t$. We will then again find a subtree with a suitably-sized solution. In this case however, we will contract the separator between the subtree and the rest of the graph to a single vertex. The next lemma shows that this does not reduce the connected vertex cover size in the subtree by more than twice the size of the separator.

\begin{lemma}\label{lem:cvc:connect-by-plus-ell}
 Let $G$ be a connected graph and let $X \subseteq V(G)$. Given a connected vertex cover $S$ of $G_X$ where $G_X$ is obtained from $X$ by identifying all vertices from $X$ into a single vertex $z$, there is a polynomial-time algorithm that finds a connected vertex cover $S'$ of size at most $|S| + 2|X|$ of $G$.
\end{lemma}
\begin{proof}
 Let $S$ be a connected vertex cover of $G_X$. Let $S'' := S \cup X\setminus \{z\}$. Observe that $S''$ is a vertex cover of $G$, such that every connected component of $G[S'']$ contains at least one vertex from $X$; thus, there are at most $|X|$ connected components. If $G[S'']$ is connected, we are done. 
 Otherwise, we show that there is a single vertex $v \in V(G)$ such that $G[S'' \cup \{v\}]$ has strictly fewer connected components than $G[S'']$. It is then straightforward to obtain $S'$ by repeatedly adding such a vertex, until $G[S'']$ is connected.  For any vertex $u \in S''$ define $C_u$ as the connected component of  vertex $u$ in $G[S'']$. 
 
 \begin{figure}
 \centering
  \includegraphics{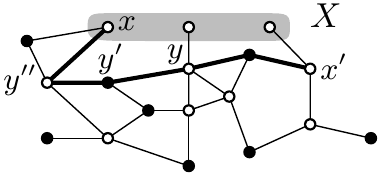}
  \caption{A graph with a vertex cover $S''$ (indicated in white) that is connected when all vertices in $X$ are identified into a single vertex. Shown are $x$, $x'$,$y$,$y'$,$y''$, and $P$ (indicated in bold) as used in the proof of Lemma~\ref{lem:cvc:connect-by-plus-ell}.}
  \label{fig:connect_VC}
 \end{figure}

  Let $x$ and $x'$ be in two distinct components in $G[S'']$, consider the shortest path $P$ from $x$ to $x'$ in $G$.  Refer to Figure~\ref{fig:connect_VC} for a sketch of the situation. By this definition, $C_x \neq C_{x'}$. Let $y$ be the first vertex in $P$ such that $y \in S''$ but $y\notin C_x$, let $y'$ be the vertex on $P$ before $y$, observe that $y' \notin S''$ since otherwise $y' \in S''$ and $y' \notin C_x$ which is a contradiction with the fact that $y$ is the first such vertex in $P$. Let $y''$ be the vertex on the path before $y'$, such that $P = (x,\ldots,y'',y',y,\ldots,x')$, where possibly $x=y''$ or $y=x'$. Observe that $y'' \in S''$ as otherwise edge $\{y'',y'\}$ is not covered, and therefore $y'' \in C_{x}$ since $y$ is the first vertex on $P$ that is in $S''$ but not in $C_x$.   Therefore, adding vertex $y'$ to $S''$ will merge connected components $C_x$ and $C_y$, such that the number of connected components in $G[S'' \cup \{y'\}]$ is strictly smaller than the number of connected components in $G[S'']$. In total, we add less than $|X|$ vertices to $S''$ obtain a connected vertex cover $S'$ and thus $|S'|\leq|S|+|X|$.
\end{proof}

We now prove the main result of this section. 

\begin{theorem}\label{thm:cvc-tk}
For every $0 < \varepsilon \leq 1$, \textsc{Connected Vertex Cover} parameterized by treewidth $\ell$ has a $(1+\varepsilon)$-approximate Turing Kernel with  $\Oh(\big(\frac{\ell^2}{\varepsilon}\big)^{\big\lceil\frac{3 + \varepsilon}{\varepsilon}\big\rceil})$ vertices.
\end{theorem}
\begin{proof}
We will use the PSAKS for \textsc{Connected Vertex Cover} from Theorem~\ref{thm:cvc-psaks}. Recall that such a PSAKS consists of a reduction algorithm $R_\mathcal{A}$ together with a solution lifting algorithm $S_\mathcal{A}$. We will use the following claim.
\begin{claim}\label{claim:cvc:obtain-approx}
 Given $0 < \delta \leq 1$ and a connected graph $G$ with tree decomposition of width $\ell$, there is a polynomial-time algorithm to determine a $d$-approximate solution for CVC or correctly decide that $\OPTcvc(G) > \frac{100\ell^2}{\delta}$, when given access to a $c$-approximate CVC-oracle that allows calls using graphs with at most $\Oh(\big(\frac{\ell^2}{\delta}\big)^{\big\lceil\frac{1 + \delta}{\delta}\big\rceil})$ vertices, where $d = \min(c\cdot (1+\delta),2)$.
\end{claim}
\begin{claimproof}
Using the fact that CVC is $2$-approximable in polynomial time~\cite{DBLP:journals/ipl/Savage82}, obtain a $2$-approximate solution $\tilde{S}$ in $G$. If $|\tilde{S}| > 200\ell^2/\delta$, return \textsc{no} and halt. Otherwise, continue by 
running $R_\mathcal{A}$ on $(G,|\tilde{S}|)$ to obtain $(G',k')$. Observe that $G'$ has at most  $\Oh(\big(\frac{\ell^2}{\delta}\big)^{\big\lceil\frac{1 + \delta}{\delta}\big\rceil})$ many vertices. Apply the $c$-approximate oracle on $G'$ to obtain CVC $S'$ in $G'$. Obtain an approximate solution $S$ in $G$ by using the solution lifting algorithm on $G'$ and $S'$. 
Output the smallest solution of $S$ and $\tilde{S}$, let this be $\hat{S}$. We show that this has the desired approximation factor, which requires an argument since the PSAKS works for $\modcvc$ instead of CVC (recall $\modcvc(G,k,S) = \min\{k+1,\text{CVC}(G,S)\}$). Observe that $|\hat{S}| \leq |\tilde{S}|$, by definition. Therefore, $|\hat{S}| \leq \modcvc(G,|\tilde{S}|,S)$. Thus 
\[\frac{|\hat{S}|}{\OPTcvc(G)} \leq \frac{\modcvc(G,|\tilde{S}|,S)}{\OPTcvc(G)} \leq\frac{\modcvc(G,|\tilde{S}|,S)}{\modOPTcvc(G,|\tilde{S}|)}.
\]By correctness of the solution lifting algorithm, we get
\[\frac{\modcvc(G,|\tilde{S}|,S)}{\modOPTcvc(G,|\tilde{S}|)} \leq (1+\delta)\frac{\modcvc(G',k',S')}{\modOPTcvc(G',k')} \leq (1+\delta)\frac{|S'|}{\OPTcvc(G')} \leq c\cdot (1+\delta),\]
by correctness of the oracle.
\end{claimproof}

\subparagraph*{Algorithm} The algorithm now proceeds as follows.  Our goal is to find a subtree of $T$ for which on the one hand, the local optimum CVC is small enough to find an approximate solution using Claim~\ref{claim:cvc:obtain-approx}, but also large enough to be able to  (among other things) add the entire set $X_t$ to the solution, without introducing a too large error. Let $\delta := \varepsilon/3$. 

For any vertex $t \in V(T)$, let $G_t$ be the graph given by $G[V_t]$ after identifying all vertices from $X_t$ into a single vertex $z_t$. Apply Claim~\ref{claim:cvc:obtain-approx} to $G$, if it returns an approximate connected vertex cover of $G$, we are done. Otherwise, $\OPTcvc(G) > \frac{100\ell^2}{\delta}$. We now aim to find a vertex $t$ such that Claim~\ref{claim:cvc:obtain-approx}  returns an approximate solution in $G_t$ of size at least $\frac{10\ell}{\delta}$.

\begin{claim}
 \label{claim:cvc:find-t}
 There is a polynomial-time algorithm that, given $G$ with tree decomposition $\mathcal{T}$ of width $\ell$ such that $\OPTcvc(G) > \frac{100\ell^2}{\delta}$, finds $t \in V(T)$ for which Claim~\ref{claim:cvc:obtain-approx} returns an approximate solution $S_t$ with $|S_t| \geq \frac{10\ell}{\delta}$, using calls to a $c$-approximate oracle of size at most $\Oh(\big(\frac{\ell^2}{\delta}\big)^{\big\lceil\frac{1 + \delta}{\delta}\big\rceil})$.
\end{claim}
\begin{claimproof}
Start with $t := r$, note that since $\OPTcvc(G) > \frac{100\ell^2}{\delta}$ and  $G_r = G $, we have that $\OPTcvc(G_r) > \frac{100\ell^2}{\delta}$, where $r$ is the root of $T$. We search through the graph maintaining $\OPTcvc(G_t) > \frac{100\ell^2}{\delta}$. 
Let $t_1,\ldots,t_m$ be the children of $t$, recall that  we may assume $m \leq 2\ell+2$ by Theorem~\ref{lem:connected-td}. For each $t_i$, apply Claim~\ref{claim:cvc:obtain-approx}. Consider the following possibilities.
\begin{itemize}
 \item There exists $i \in [m]$ such that the claim determines $\OPTcvc(G_{t_i}) > \frac{100\ell^2}{\delta}$, in this case, recurse with this $t_i$.
 \item There exists $i \in [m]$ such that the claim returns a $\min\{2,(1+\delta)\cdot c\}$-approximate solution $S_{t_i}$ of size at least $\frac{10\ell}{\delta}$ for CVC. In this case, return $t := t_i$.
 \item Otherwise. Thus, for every $i \in [m]$, the algorithm returns a connected vertex cover $S_i$ of size at most $\frac{10\ell}{\delta}$ for CVC in $G_{t_i}$. Obtain a connected vertex cover $S_i'$ of $G[V_{t_i}]$ of size at most $|S_i| + 2(\ell+1)$ using Lemma~\ref{lem:cvc:connect-by-plus-ell}. We will argue that in this case $CVC(G_t) < \frac{55\ell^2}{\delta}$, which is a contradiction. We obtain a connected vertex cover of $G_t$ as follows. Let $\hat{S}_t := \bigcup_{i \in [m]} (S_i') \cup \{z_t\}$. Observe that $\hat{S}_t$ has size at most $(2\ell+2)\cdot\frac{13\ell}{\delta} + 1\leq \frac{55\ell^2}{\delta}$. It is easy to observe that $\hat{S}_t$ is indeed a connected vertex cover of $G_t$.
\end{itemize}
Observe that from the steps above, we always get a connected vertex cover $S_t$ of $G_t$, that is a $(1+\delta)\cdot c$-approximation of $\OPTcvc(G_t)$ and has size at least $\frac{10\ell}{\delta}$. 
\end{claimproof}

Using Claim~\ref{claim:cvc:find-t}, we obtain a node $t$ and a connected vertex cover $S_t$ of $G_t$, that is a $\min\{(1+\delta)\cdot c,2\}$-approximation of $\OPTcvc(G_t)$ and has size at least $\frac{10\ell}{\delta}$. Use Lemma~\ref{lem:cvc:connect-by-plus-ell} to  obtain a connected vertex cover $S_t'$ of $G[V_t]$ of size at most  $|S_t| + 2(\ell + 1)$, containing $X_t$. 

We now obtain graph $G'$ by removing all vertices in $V_t\setminus X_t$ from $G$ and then contracting all vertices in $X_t$ to a single vertex $z_t$. Let $\mathcal{T}'$ to be a tree decomposition of $G'$, one may obtain $\mathcal{T'}$ by  replacing occurrences of vertices in $V_t$ by $z_t$ in $\mathcal{T}$. 
Since $G'$ is strictly smaller than $G$, we may use the algorithm described above to obtain a $c\cdot(1+\varepsilon)$-approximate solution $S'$ for $\OPTcvc(G')$, using $\mathcal{T}'$.
Output $S:= S'\cup S_t' \setminus \{z_t\}$. 

\subparagraph*{Correctness} We start by showing that $S$ is a connected vertex cover. Verify that it is indeed a vertex cover of $G$: any edge within $G'$ is covered as $S' \subseteq S$, any edge in $G_t$ is covered since $S_t'\subseteq S$ and any other edge has at least one endpoint in $X_t \subseteq S$ and is thereby covered. It remains to verify that $G[S]$ is connected. Clearly, $G[V_t \cap S]$ is connected since it corresponds to $G[S_t']$.   Let $\tilde{G} := G - (V_t \setminus X_t)$. We show that every connected component of $\tilde{G}[S]$ contains at least one vertex from $X_t$, such that the entire graph is connected as $X_t \subseteq S$ and the vertices in $X_t$ are in the same connected component as observed earlier. Suppose not,  let $C$ be such a component not containing any vertex in $X_t$.
Consider $G'[S']$. Observe that $C$ is also a connected component of $G'[S']$. Furthermore, vertex $ z_t $ is not adjacent to any vertex in $C$, as otherwise there is an edge from some vertex in $C$ to some vertex in $X_t$ in $\tilde{G}$, since $X_t\subseteq S$ this contradicts that $C$ contains no vertex from $X_t$. Since $G'$ is connected however, $z_t$ has an incident edge $\{z_t,u\}$ for some $u \in V(G')$ and thus $u \in S'$ or $z_t \in S'$. In both cases there is a vertex in $S'$ that is not in connected component $C$, a contradiction with the assumption that $S'$ is a connected vertex cover of $G'$.

We now show that we indeed achieve the desired approximation factor. 
\begin{claim}\label{claim:cvc:size}
 $|S| \leq c \cdot (1+\varepsilon) \cdot \OPTcvc(G)$
\end{claim}
\begin{claimproof}
Let $S^*$ be a minimum connected vertex cover of $G$. Assume for now  $|S^* \cap V(G_t)| \geq 4/\delta$.
 \begin{align*}
  |S| &\leq |S_t'| + |S'| \\
  &\leq |S_t| + 2(\ell + 1) + c \cdot (1+\varepsilon)\OPTcvc(G')\\
  \intertext{Using $|S_t| \geq \frac{10\ell}{\delta}$}
  &\leq |S_t| + \frac{\delta}{2}|S_t| + c \cdot (1+\varepsilon)\OPTcvc(G')\\
  &\leq c \cdot (1+\delta)(1+\delta/2) \OPTcvc(G_t)+ c \cdot (1+\varepsilon)|(S^* \cap V(G'))\cup \{z_t\}|\\
  &\leq c \cdot (1+\delta)(1+\delta/2) |(S^*\cap V(G_t))\cup \{z_t\}|  + c \cdot (1+\varepsilon)|(S^* \cap  V(G'))\cup \{z_t\}|\\ 
   &\leq c \cdot (1+\delta)(1+\delta/2) (|S^*\cap V(G_t)|+1)  + c \cdot (1+\varepsilon)|(S^* \cap  V(G'))\cup \{z_t\}|\\ 
      \intertext{By assuming $|S^* \cap V(G_t)| \geq 4/\delta$, and then using $\delta = \varepsilon/3$}
      &\leq c \cdot (1+\delta)(1+\delta/2) (1+\delta/4)(|S^*\cap V(G_t)|)  + c \cdot (1+\varepsilon)|(S^* \cap  V(G'))\cup \{z_t\}|\\ 
      &\leq c \cdot (1+\varepsilon)(|S^*\cap V(G_t)|)  + c \cdot (1+\varepsilon)|(S^* \cap  V(G'))\cup \{z_t\}|\\
    \intertext{Observe that since $G_t$ and $G'$ are non-empty, $S^*$ must contain a vertex from $X_t$}
      &\leq c \cdot (1+\varepsilon) |S^*| = c \cdot (1+\varepsilon) \cdot \OPTcvc(G).
 \end{align*}
 It remains to observe that $|S^* \cap V(G_t)| \geq 4/\delta$ is a reasonable assumption. Suppose not, then $\OPTcvc(G_t) \leq |S^* \cap V(G_t)| + 1 \leq 4/\delta + 1$. However, $S_t \geq \frac{10\ell}{\delta}\geq 2 \cdot \OPTcvc(G_t)$, meaning that $S_t$ is not a $2$-approximation in $G_t$, which is a contradiction.
\end{claimproof}

Having shown the correctness of the procedure, it remains to argue the size of this Turing kernel. Observe that the oracle is only used when applying Claim~\ref{claim:cvc:obtain-approx}. As such, we may bound the size of the kernel by
 $\Oh(\big(\frac{\ell^2}{\delta}\big)^{\big\lceil\frac{1 + \delta}{\delta}\big\rceil}) =\Oh(\big(\frac{\ell^2}{\varepsilon}\big)^{\big\lceil\frac{3 + \varepsilon}{\varepsilon}\big\rceil})$, recall that $\delta = \frac{\varepsilon}{3}$.
\end{proof}

\section{Meta result}
\label{sec:general}
\newcommand{\ph}{\ensuremath{\varphi}\xspace}

In this section we will describe a wide range of graph problems for which approximate Turing kernels can be obtained. The problems we will consider satisfy certain additional constraints, such that the general strategy already described for the \textsc{Vertex Cover} problem can be applied. Informally speaking, we need the following requirements. First of all, the problems should behave nicely with respect to taking the disjoint union of graphs. Secondly, we want to look at what happens for induced subgraphs. We will only consider problems whose value cannot increase when taking an induced subgraph. Furthermore, we restrict how much the optimal value can decrease when taking an induced subgraph. Finally, we require existence of a PSAKS and an approximation algorithm for the problem. We use the following definitions.

\begin{definition}
 Let $\ph \colon \mathbb{R}\times \mathbb{N}\to\mathbb{R}$ be a function. A \emph{$\varphi$-approximation algorithm} for a problem $\P$ is a  polynomial-time algorithm that, given an instance $G$ with tree decomposition $\mathcal{T}$ of width $\ell$, outputs a solution $S$ such that (for minimization problems) $\P(G,S) \leq \ph(\OPTp(G),\ell)$, and (for maximization problems) $\ph(\P(G,S),\ell) \geq \OPTp(G)$. 
\end{definition}
To illustrate this definition, observe that since \textsc{Vertex Cover} has a $2$-approximation, this same approximation algorithm serves as a $\varphi$-approximation with $\varphi(s,\ell) = 2s$.  
We use the above definition to allow the approximation factor of the algorithm to depend on the size of the optimal solution and the treewidth of the considered graph.

We can now formally define our notion of a \emph{friendly} problem.

\begin{definition}
Let \P be an optimization problem whose input is a graph. We will say that it is \emph{\friendly} if it satisfies the following conditions.
\begin{enumerate}
 \item \label{nice:union} For all graphs $G$, $G_1$, and $G_2$ such that $G$ is the disjoint union of graphs $G_1$ and $G_2$,  $\OPTp(G) = \OPTp(G_1) + \OPTp(G_2)$. In particular, if $S_1$ is a solution for $G_1$ and $S_2$ is a solution for $G_2$, then $S_1\cup S_2$ is a solution for $G$ and \[\P(G, S_1 \cup S_2) = \P(G_1,S_1) + \P(G_2,S_2).\] In the other direction, given solution $S$ in $G$ it can efficiently be split into solutions $S_1$ in $G_1$ and $S_2$ in $G_2$ satisfying the above. For consistency, we require that the size of the optimal solution in the empty graph is zero.
\item \label{nice:induced} There exists a non-decreasing polynomial function $f$ such that for all graphs $G$, for all $X\subseteq V(G)$: \[\OPTp(G) \leq \OPTp(G-X) + f(|X|)\text{, and } \OPTp(G-X)\leq \OPTp(G).\]
In particular, for minimization problems there is a polynomial-time algorithm $\mathcal{A}$ that, given a solution $S'$ in $G-X$, outputs a solution $S$ for $G$ such that $\P(G,S) \leq \P(G-X,S') + f(|X|)$. For maximization problems we require that any solution $S$ for $G-X$ is also a solution for $G$ and $\P(G,S) = \P(G-X,S)$.
\item \label{nice:psaks} \modP parameterized by $k + \ell$, where $k$ is the solution value and $\ell$ is the treewidth, has a $(1+\delta)$-approximate kernel for all $\delta > 0$, that has $h(\delta,k + \ell)$ vertices for some  function $h$ that is polynomial in its second parameter.
\item \label{nice:approx} \P has a $\ph$-approximation algorithm 
for some polynomial function $\ph$ such that $\alpha\cdot \varphi(k,\ell) < \varphi(\alpha\cdot k,\ell)$ for all $\alpha > 1$, and $\varphi$ is non-decreasing in its first parameter.  
\end{enumerate}
\end{definition}
Observe that many well-known vertex subset problems fit in this framework. As an example, let us verify them for the \VC problem. The first point is immediate. For the second point, let $\mathcal{A}(G,X,S)$ output $S' := S \cup X$. Verify that indeed this satisfies the conditions with $f(|X|) = |X|$. The third point follows with some extra work from the fact that \VC has a kernel with $2k$ vertices, this kernel can then be shown to be $1$-approximate. For the last point, it is well-known that \VC has a $2$-approximation algorithm.

The next lemma will be used in a similar way as Claim~\ref{claim:vc:find-t} was used for \textsc{Vertex Cover}.
It shows how to find a suitable node $t$ of the tree decomposition,  such that we may split the graph at this node and recurse. Furthermore it gives an approximate solution for the subtree rooted at $t$.

%TODO lemma for maximization problems?
\begin{lemma}\label{lem:find_t_minimization}
  Let $\P$ be a \friendly graph optimization problem. There is a polynomial-time algorithm $\mathcal{B}$ with access to a $c$-approximate oracle. It takes as input a graph $G$ with nice tree decomposition $\mathcal{T}$ of width $\ell$ and a number $0 < \delta \leq 1$, and outputs either
  \begin{itemize}
  \item a node $t$  such that $\OPTp(G[V_t\setminus X_t]) \geq \frac{f(\ell + 1)}{\delta}$ together with a $(c \cdot (1+\delta))$-approximate solution $S_t$ to \P in $G[V_t \setminus X_t]$, or
  \item  a $c\cdot (1+\delta)$-approximate solution for $G$,
  \end{itemize}
using calls to the oracle on graphs with at most $h(\delta, \ph(k,\ell)+\ell)$ vertices, where $k = \frac{2f(\ell+1)}{\delta} + f(1)$.
\end{lemma}
We will prove the result separately for maximization and minimization problems.

\begin{proof}[Proof of Lemma~\ref{lem:find_t_minimization}: Maximization problems]
 Let $r$ be the root of $\mathcal{T}$, and observe that $G = G[V_r \setminus X_r]$ since $X_r = \emptyset$. Let $k:=\frac{2f(\ell+1)}{\delta} + f(1)$. 
 Compute a $\ph$-approximate solution $\tilde{S}$ in $G$. We do a case distinction on the value of this solution.
 
 %Case small
 If $\P(G,\tilde{S}) \leq k$, then 
 apply the PSAKS with approximation ratio $1+\delta$ to $(G,\ph(k,\ell) + \ell)$ and obtain instance $(G',k')$ with at most $h(\delta, \ph(k,\ell)+\ell)$ vertices. Obtain solution $S'$ by applying the  $c$-approximate oracle on $G'$. Apply the solution lifting algorithm to $S'$ to obtain a solution $S$ for $G$. We start by showing that $S$ is the desired approximate solution. 
 Clearly, $\P(G',S') \geq \frac{1}{c} \cdot \OPTp(G')$ by correctness of the oracle. If $\P(G',S') > k'$, then $\modP(G',k',S') = k' + 1$ and thus $\modP(G',k',S') \geq \modOPTp(G',k')$. Otherwise, we have 
  $\modP(G',k',S') = \P(G',S') \geq \frac{1}{c} \cdot \OPTp(G') \geq \frac{1}{c} \cdot \modOPTp(G',k')$. From the properties of the solution lifting algorithm, it now follows that  
 $\modP(G,\ph(k,\ell) + \ell,S) \geq \frac{1}{c(1+\delta)} \modOPTp(G,\ph(k,\ell) + \ell)$. %It follows from $\P(G,\tilde{S}) \leq k$ that $\OPTp(G) \leq \ph(\P(G,\tilde{S}),\ell)\leq \ph(k,\ell)$ and thus $\modOPTp(G,\ph(k,\ell) + \ell)=\OPTp(G)$. 
 Observe that since $\P(G,\tilde{S}) \leq k$ and $\varphi$ non-decreasing in its first parameter, we get that $\OPTp(G) \leq \varphi(\P(G,\tilde{S}),\ell) \leq \varphi(k,\ell)$ and thereby $\OPTp(G) = \modOPTp(G,\varphi(k,\ell) + \ell)$. 
 It follows that  $\P(G,S) \geq \modP(G,\ph(k,\ell) + \ell,S)\geq\frac{1}{c(1+\delta)} \modOPTp(G,\ph(k,\ell) + \ell)= \frac{1}{c(\delta + 1)}\OPTp(G)$.%, as desired.

 %Case large
 Suppose $\P(G,\tilde{S}) > k$. For every node $t \in T$, compute a $\ph$-approximate solution $\tilde{S}_t$ for graph $G[V_t\setminus X_t]$. We start by showing how to find a node $t \in V(T)$ such that both $\P(G[V_t\setminus X_t],\tilde{S}_t) \leq k$, and  $\OPTp(G[V_t \setminus X_t]) \geq \frac{f(\ell + 1)}{\delta}$. Start by observing that for the leaf vertices, it holds that $\P(G[V_t\setminus X_t],\tilde{S}_t) = 0 \leq k$. On the other hand, for the root, we found that $\P(G[V_r\setminus X_r],\tilde{S}_r) = \P(G,\tilde{S}) > k$. As such, we can find a node $p$  such that  $\P(G[V_p\setminus X_p],\tilde{S}_p) > k$, while for all of its children $t$ it holds that $\P(G[V_t\setminus X_t],\tilde{S}_t) \leq k$. We show that one of the children of $p$ has the desired properties.  The result that $\P(G[V_t\setminus X_t],\tilde{S}_t) \leq k$ for all children of $p$ is immediate. On the other hand, observe that $\OPTp(G[V_p\setminus X_p]) \geq \P(G[V_p\setminus X_p],\tilde{S}_p) \geq k\geq \frac{2f(\ell+1)}{\delta}$, by assumption. We do a case distinction on the type of node that $p$ is in the nice tree decomposition. 
 \begin{itemize}
  \item $p$ is an introduce or forget node. In this case, $p$ has exactly one child $t$ and $V_t\setminus X_t = V_p\setminus X_p$, or $V_t \setminus X_t = (V_p\setminus X_p)\setminus \{v\}$ for some $v \in V(G)$.  Since \P is friendly, we get that $\OPTp(G[V_t \setminus X_t]) \geq \OPTp(G[V_p \setminus X_p]) - f(1) \geq \frac{f(\ell + 1)}{\delta}$.
 \item $p$ is a join node. In this case, $p$ has exactly two children $t_1$ and $t_2$ and $G[V_p \setminus X_p]$ is the disjoint union of $G[V_{t_1} \setminus X_{t_1}]$ and $G[V_{t_2} \setminus X_{t_2}]$.
Obtain $S_1$ and $S_2$ such that $\tilde{S}_p = S_1 \cup S_2$ and $S_1$ is a solution in $G[V_{t_1} \setminus X_{t_1}]$, $S_2$ in  $G[V_{t_2} \setminus X_{t_2}]$, and $\P(G[V_p\setminus X_p],\tilde{S}_p) = \P(G[V_{t_1}\setminus X_{t_1}],S_1) + \P(G[V_{t_2}\setminus X_{t_2}],S_2)$. This can be done since $\mathcal{P}$ is \friendly.
  
 Therefore, there is $i \in [2]$ such that $\OPTp(G[V_{t_i}\setminus X_{t_i}]) \geq \P(G[V_{t_i}\setminus X_{t_i}], S_i)\geq \P(G[V_p\setminus X_p],\tilde{S}_p)/2\geq \frac{f(\ell + 1)}{\delta}$.
 \end{itemize}
 So, we have obtained a node $t$ such that $\P(G[V_t\setminus X_t],\tilde{S}_t) \leq k$, and  $\OPTp(G[V_t \setminus X_t]) \geq \frac{f(\ell + 1)}{\delta}$. We now show how to obtain $S_t$. Apply the PSAKS with ratio $1+\delta$ to $(G[V_t\setminus X_t], \ph(k,\ell) + \ell)$ and obtain instance $(G',k')$. Apply the $c$-approximate oracle on $G'$ to obtain a solution $S''$. Apply the solution lifting algorithm to $S''$ to obtain solution $S_t$ in $G[V_t\setminus X_t]$. With similar arguments as before, $S_t$ is a $c(1+\delta)$-approximate solution in $G[V_t\setminus X_t]$. Output $t$ and $S_t$.
\end{proof}

\begin{proof}[Proof of Lemma~\ref{lem:find_t_minimization}: Minimization problems]
%\todo{Fix issue with $\varphi$ has now two vars} 
 Let $r$ be the root of $\mathcal{T}$, and observe that $G = G[V_r \setminus X_r]$ since $X_r = \emptyset$. Let $k:=\frac{2f(\ell+1)}{\delta} + f(1)$. 
 Compute a $\ph$-approximate solution $\tilde{S}$ in $G$. We do a case distinction on the value of this solution.
 
 %Case small
 If $\P(G,\tilde{S}) \leq \ph(k,\ell)$, then 
 apply the PSAKS with with approximation ratio $1+\delta$ to $(G,\ph(k,\ell) + \ell)$ and obtain instance $(G',k')$. Obtain solution $S'$ by applying the $c$-approximate oracle on $G'$. Apply the solution lifting algorithm to $S'$ to obtain a solution $S$ for $G$, output the best of $\tilde{S}$ and $S$. We start by showing that this is correct. Clearly, $\P(G',S') \leq c \cdot \OPTp(G')$. If $\OPTp(G') > k'$, then  $\modP(G',k',S') = \modOPTp(G',k') =k' +1$. Otherwise, $\P(G',S') \geq \modP(G',k',S')$ and $\modOPTp(G',k') = \OPTp(G')$, such that 
 $\modP(G',k',S') \leq \P(G',S') \leq c \cdot \modOPTp(G',k')$. In both cases, 
we thus get that after applying the solution lifting algorithm we have  $\modP(G,\ph(k,\ell) + \ell,S) \leq c \cdot (1+\delta) \modOPTp(G,\ph(k,\ell) + \ell)$. If $\modP(G,\ph(k,\ell) + \ell,S) \leq \ph(k,\ell)+\ell$ this implies $\modP(G,\ph(k,\ell) + \ell,S)=\P(G,S)$ and indeed 
 $\P(G,S) \leq c \cdot (1+\delta) \OPTp(G)$, as desired. Otherwise, we see that $\P(G,\tilde{S})\leq \varphi(k,\ell) < \modP(G,\ph(k,\ell) + \ell,S) \leq c \cdot (1+\delta) \modOPTp(G,\ph(k,\ell) + \ell) \leq c \cdot (1+\delta) \OPTp(G)$, concluding this case.
 
 %Case large
 Suppose $\P(G,\tilde{S}) > \ph(k,\ell)$. For every node $t$ in tree decomposition $\mathcal{T}$, compute a $\ph$-approximate solution $\tilde{S}_t$ for graph $G[V_t\setminus X_t]$. We start by showing that there exists a $t \in V(T)$ such that on the one hand $\P(G[V_t\setminus X_t],\tilde{S}_t) \leq \ph(k,\ell)$, and on the other hand $\OPTp(G[V_t \setminus X_t]) \geq \frac{f(\ell + 1)}{\delta}$. Start by observing that for the leaf vertices, it holds that $\P(G[V_t\setminus X_t],\tilde{S}_t) = 0 \leq \ph(k,\ell)$. On the other hand, for the root, we found that $\P(G[V_r\setminus X_r],\tilde{S}_r) = \P(G,\tilde{S}) > \ph(k,\ell)$. As such, we can find a node $p$  such that  $\P(G[V_p\setminus X_p],\tilde{S}_p) > \ph(k,\ell)$, while for all of its children $t$ it holds that $\P(G[V_t\setminus X_t],\tilde{S}_t) \leq \ph(k,\ell)$. We show that one of the children of $p$, or $p$ itself for a different choice of $\tilde{S}_t$, has the desired properties.  The result that $\P(G[V_t\setminus X_t],\tilde{S}_t) \leq \ph(k,\ell)$ for all children of $p$ is immediate. On the other hand, observe that $\OPTp(G[V_p\setminus X_p]) \geq k$. We do a case distinction on the type of node that $p$ is in the nice tree decomposition.
 \begin{itemize}
  \item $p$ is an introduce or forget node. In this case, $p$ has exactly one child $t$ and $\OPTp(G[V_t \setminus X_t]) \geq \OPTp(G[V_p \setminus X_p]) - f(1) \geq k - f(1) \geq \frac{f(\ell + 1)}{\delta}$.
 \item $p$ is a join node. In this case, $p$ has exactly two children $t_1$ and $t_2$. If $\P(G[V_{t_i} \setminus X_{t_i}],\tilde{S}_{t_i})\leq \varphi(k,\ell)/2$ for $i\in[2]$, then $\P(G[V_p \setminus X_p],\tilde{S}_{t_1}\cup \tilde{S}_{t_2})\leq \varphi(k,\ell)$ and we may use $p$ as the desired node with solution $\tilde{S}_p := \tilde{S}_{t_1}\cup \tilde{S}_{t_2}$. Otherwise, there is $i\in[2]$ such that $\P(G[V_{t_i} \setminus X_{t_i}],\tilde{S}_{t_i})> \varphi(k,\ell)/2 \geq \varphi(\frac{k}{2},\ell)$, implying $\OPTp(G[V_{t_i} \setminus X_{t_i}]) \geq \frac{k}{2} \geq \frac{f(\ell+1)}{\delta}$.

 %$p$ is a join node. In this case, $p$ has exactly two children $t_1$ and $t_2$ and $\OPTp(G[V_p \setminus X_p]) =\OPTp(G[V_{t_1} \setminus X_{t_1}]) + \OPTp(G[V_{t_2} \setminus X_{t_2}])$ since $\P$ is \friendly. Therefore, there is a child $t$ of $p$ such that $\OPTp(G[V_t \setminus X_t]) \geq \OPTp(G[V_p \setminus X_p]) /2 \geq \frac{f(\ell + 1)}{\delta}$.
 \end{itemize}
 So, we have obtained a node $t$ such that $\P(G[V_t\setminus X_t],\tilde{S}_t) \leq \ph(k,\ell)$, and  $\OPTp(G[V_t \setminus X_t]) \geq \frac{f(\ell + 1)}{\delta}$. We now show how to obtain $S_t$. Apply the PSAKS with ratio $1+\delta$ to $(G[V_t\setminus X_t], \ph(k,\ell) + \ell)$ and obtain instance $(G',k')$. Apply the $c$-approximate oracle on $G'$ to obtain a solution $S''$. Apply the solution lifting algorithm to $S''$ to obtain a solution $S'$ in $G[V_t\setminus X_t]$. Let $S_t$ be the best solution out of $S'$ and $\tilde{S}_t$. With similar arguments as before, $S_t$ is a $c(1+\delta)$-approximate solution in $G[V_t\setminus X_t]$. Output $t$ and $S_t$.
\end{proof}

The next theorem gives a polynomial-size $(1+\varepsilon)$-approximate Turing kernel with parameter treewidth for any \friendly optimization problem \P. The Turing kernel follows the same ideas as the Turing kernels presented in the remainder of this paper, using Lemma~\ref{lem:find_t_minimization} to find a node in the tree decomposition where we can split the graph.
\begin{theorem}\label{thm:TK}
Let $\P$ be a \friendly optimization problem on graphs. Then \P parameterized by treewidth has a $(1+\varepsilon)$-approximate Turing kernel with $h(\frac{\varepsilon}{3}, \ph(\frac{6f(\ell+1)}{\varepsilon} + f(1),\ell)+\ell)$ vertices, for all $0 < \varepsilon \leq 1$. 
\end{theorem}
While the description of the Turing kernel is mostly the same for maximization and minimization problems (refer to Algorithm~\ref{alg:approx-general-min}), the correctness proof will differ quite significantly. Therefore, these cases will be proven separately. %The proof for minimization problems is deferred to Appendix~\ref{app:sec:general}.
\begin{algorithm}
\caption{An approximate Turing kernel for \friendly optimization problems \P.}
\label{alg:approx-general-min}
\begin{algorithmic}[1]
\Procedure{$\textsc{ApproxP}(G,\mathcal{T},\varepsilon)$}{}
\State Turn $\mathcal{T}$ into a nice tree decomposition
\State Apply Lemma~\ref{lem:find_t_minimization} for $\delta:= \varepsilon/3$
\If {this outputs an approximate solution $S$ for $G$}
    \State \label{step:return_1}\Return $S$
\Else
     \label{step:P2:find_t}\emph{ // We obtained $t \in V(T)$, $c(1+\delta)$-approximate solution $S_t$ for \P in  $G[V_t\setminus X_t]$ } 
    
    \qquad \quad \emph{such that $\OPTp(G[V_t\setminus X_t]) \geq \frac{f(\ell + 1)}{\delta}$ }
    \State Let $G' := G - V_t$.
    \State Obtain $\mathcal{T}'$ from $\mathcal{T}$ by removing the subtree rooted at $t$ and all vertices in $X_t$
    \State Let $S' := \textsc{ApproxP}(G',\mathcal{T}',\varepsilon)$
    \State \label{step:return_2} \Return $S:= \mathcal{A}(G, X_t, S' \cup S_t)$
\EndIf
\EndProcedure
\end{algorithmic}
\end{algorithm}

\begin{proof}[Proof of Theorem~\ref{thm:TK}: Minimization problems]\label{proof:thm:TK}
 Let a graph $G$ with tree decomposition $\mathcal{T}$ be given. We show that Algorithm~\ref{alg:approx-general-min} is the desired approximate Turing kernel. It is easy to verify that all calls to the oracle have size at most $h(\frac{\varepsilon}{3}, \varphi(\frac{2f(\ell+1)}{\delta} + f(1),\ell)+\ell)$.   If we return a set $S$ in Step~\ref{step:return_1}, it is immediate from the correctness of Lemma~\ref{lem:find_t_minimization} that the returned solution is correct and has the right approximation ratio. Otherwise, a set is returned in Step~\ref{step:return_2}. Observe that in this case, $S_t$ is a solution for $G[V_t\setminus X_t]$ and $S'$ is a solution in $G-V_t$. Observe that the vertices $V_t \setminus X_t$ and $V(G) \setminus V_t$ are not in the same connected component in $G-X_t$. As such, since \P is \friendly, we obtain that $S_t \cup S'$ is a solution for \P in $G' := G - X_t$. Therefore, the returned solution $S$ is a solution for \P in $G$ of size at most $\P(G-X_t,S'\cup S_t) + f(|X_t|) \leq \P(G-X_t,S'\cup S_t) + f(\ell + 1)$. It remains to argue that $S$ has the desired value.
  \begin{align*}
  \P(G,S) &\leq \P(G-X_t, S' \cup S_t) + f(\ell + 1)\\
  &= \P(G-V_t, S') + \P(G[V_t\setminus X_t], S_t) + f(\ell + 1)\\
  &\leq c \cdot (1+\varepsilon) \cdot \OPTp(G-V_t) + (1+\delta) \cdot c \cdot \OPTp(G[V_t\setminus X_t]) + f(\ell + 1)\\
  &\leq c \cdot (1+\varepsilon) \cdot \OPTp(G-V_t) + (1+\delta) \cdot c \cdot \OPTp(G[V_t\setminus X_t])\\&\qquad + \delta \cdot \OPTp(G[V_t\setminus X_t])\\
  &\leq c \cdot (1+\varepsilon) \cdot (\OPTp(G-V_t) + \OPTp(G[V_t\setminus X_t]))\\
  &= c \cdot (1+\varepsilon) \cdot \OPTp(G-X_t) \leq c \cdot (1+\varepsilon) \cdot \OPTp(G).\qedhere
  \end{align*}
  \end{proof}
\begin{proof}[Proof of Theorem~\ref{thm:TK}: Maximization problems] Let \P be a \friendly maximization problem. We show that Algorithm~\ref{alg:approx-general-min} is the desired approximate Turing kernel, where we let $\mathcal{A}(G,X_t,S' \cup S_t)$ return $S' \cup S_t$.

It is easy to see that since \P is friendly, the algorithm indeed returns a correct solution for \P in $G$, it remains to prove the size bound.

 \begin{align*}
  \OPTp(G) &\leq \OPTp(G-X_t) + f(\ell + 1)\\
  &= \OPTp(G-V_t) + \OPTp(G[V_t\setminus X_t]) + f(\ell + 1)\\
  &\leq \OPTp(G-V_t) + (1+\delta) \cdot \OPTp(G[V_t\setminus X_t])\\
  &\leq c \cdot (1+\varepsilon)\cdot \P(G-V_t,S') + c \cdot (1+\delta)^2 \cdot \P(G[V_t \setminus X_t], S_t)\\
  &\leq c \cdot (1+\varepsilon) \cdot ( \P(G-V_t,S') + \P(G[V_t \setminus X_t], S_t))\\
  &= c \cdot (1+\varepsilon) \cdot ( \P(G-X_t, S'\cup S_t)) = c \cdot (1+\varepsilon) \cdot ( \P(G, S'\cup S_t)).\qedhere
 \end{align*}
\end{proof}

\subsection{Consequences}

We show that a number of well-known graph problems are \friendly in the next lemma.
\begin{lemma}\label{lem:general-result}
 The following problems are \friendly (with respect to the following bounds).
 \begin{itemize}
  \item \textsc{Independent Set} with $f(x) = x$, $h(\delta,m) = (m+1)^2$,  $\ph(s,\ell) = (\ell + 1)\cdot s$.
  \item \textsc{Vertex-Disjoint $H$-packing} for connected graphs $H$, with $|V(H)|$ constant, with $f(x) = x$, $h(\delta,k) = \Oh(k^{|V(H)|-1})$, $\ph(s,\ell) = |V(H)| \cdot s$.
  \item \textsc{Vertex Cover} with $f(x) = x$, $h(\delta,k) = 2k$, $\ph(s,\ell) = 2s$.
  \item \textsc{Clique Cover} with $f(x) = x$, $h(\delta,m) = m(m+1)$, $\ph(s,\ell) = (\ell+1) \cdot s$
  \item \textsc{Feedback Vertex Set} with $f(x)=x$, $h(\delta,k) = 4k^2$,  $\ph(s,\ell) = 2s$.
  \item \textsc{Edge Dominating Set} with $f(x) = x$,  $h(\delta,k) = 4k^2  +4k$, $\varphi(s,\ell) = 2s$.
 \end{itemize}
\end{lemma}
\begin{proof}%[{Proof of Lemma \ref{lem:general-result}}]
\ 
 \begin{description}

  \item[Independent Set] Clearly, if $G$ is the disjoint union of two graphs $G_1$ and $G_2$, then the union of an independent set in $G_1$ and an independent set in $G_2$ forms an independent set in $G$. Conversely, restricting an independent set in $G$ to $V(G_1)$ (respectively $V(G_2)$) results in an independent set in $G_1$ (respectively, $G_2$).  Furthermore, if $X$ is a subset of $G$ it is easy to verify that $\optis(G) \leq \optis(G-X) + |X|$ and that $\optis(G-X) \leq \optis(G)$ as any independent set in $G-X$ is an independent set in $G$. The PSAKS parameterized by $m := k + \ell$ is as follows. It is known that any graph of treewidth $\ell$ has an independent set of size at least $|V(G)|/(\ell+1)$. This can be seen from the fact that such graphs are $\ell$-degenerate, meaning that there is an order of the vertices $v_1,\ldots,v_n$ such that $v_i$ has degree at most $\ell$ in $G[v_1,\ldots,v_i]$. As such, an independent set of size $|V(G)|/(\ell+1)$ can be greedily constructed.  
  
   Thus, if $|V(G)| > (m+1)^2$, we simply let $G'$ be the graph consisting of an independent set of size $m + 1$. The solution lifting algorithm can then simply find a size-$(m+1)$ independent set and output it. This is always an optimal solution for \modP, since it does not distinguish between solutions of size larger than $m$.  Otherwise, we obtain that $|V(G)| \leq (m+1)^2$ and the PSAKS will not modify $G$. In both cases, we output a graph on at most $(m+1)^2$ vertices.
  
  It remains to show that there is an approximation algorithm, the idea is equivalent to the PSAKS. Return an independent set in $G$ of size at least $|V(G)|/(\ell+1)$. Then indeed $\varphi(|V(G)|/(\ell+1),\ell) =  |V(G)| \geq \optis(G)$.
  
  \item[Vertex-Disjoint $H$-Packing] Requirements~\ref{nice:union} and~\ref{nice:induced} are easily verified for $f(|X|) = |X|$, as any vertex in $X$ could be contained in at most one graph in any copy of $H$. 
  
  A simple approximation algorithm for \textsc{Vertex-Disjoint $H$-Packing} is to simply return any maximal $H$-packing $S$. We show that $|S| \geq \frac{1}{|V(H)|}\OPTp(G)$, such that this is an $\varphi$-approximation algorithm  with $\varphi(s,\ell) = |V(H)| \cdot s$. Suppose there is an optimal solution $S^*$ with $|S^*| > |V(H)|\cdot |S|$. Since the copies of $H$ in $S$  are vertex-disjoint, $S$ uses exactly $|V(H)| \cdot |S|$ vertices. Since $S^*$ contains more than $|V(H)| \cdot |S|$ elements, it follows that there is $s \in S^*$ that uses no vertices used by $S$, contradicting that $S$ is maximal. %\todo{find citation}.
  
  The existence of a PSAKS is shown in Lemma~\ref{lem:packing:approxkernel}.

  \item[Vertex Cover] Requirements~\ref{nice:union} and \ref{nice:induced} are easily verified for vertex cover, let algorithm $\mathcal{A}$ simply output the union of the given solution with set $X$. As (implicitly) observed in the proof of Lemma~\ref{lem:approximate-vc}, \textsc{Vertex Cover} has a $1$-approximate kernel of size $2k$. Furthermore, it is well-known to be $2$-approximable.
  \item[Clique Cover] Requirement~\ref{nice:union} is easy to verify. We show Requirement~\ref{nice:induced}. Let $X \subseteq V(G)$. Let $S$ be a clique cover of $G$, it is easy to see that $\{s\setminus X \mid s \in S\}$ is a clique cover of $G-X$, of size at most $|S|$. Therefore, $\OPTp(G) \geq \OPTp(G-X)$. Furthermore, let algorithm $\mathcal{A}$ when given $G$, clique cover $S$ of $G-X$ and $X$ output the clique cover $S \cup \{\{x\} \mid x \in X\}$. Then this is a clique  cover of $G$ and it has size at most $|S| + |X| \leq |S| + f(|X|)$.
  
  To show Requirement~\ref{nice:psaks}, we obtain a $1$-approximate kernel for \textsc{Clique Cover} in a somewhat similar way as for \textsc{Independent Set}. Observe that any $n$-vertex graph with treewidth $\ell$ has a minimum clique cover of size at least $\frac{n}{\ell+1}$. So, given $G$ and parameter $m := k + \ell$, if $n > m(m+1) \geq k \cdot (\ell + 1)$, we know for sure that $G$ does not have a minimum  clique cover of size $k$. The reduction algorithm reduces $G$ to an independent set of size $m+1$. The solution lifting algorithm (irrespective of the solution given for $G'$) outputs $V(G)$. Otherwise, if $n \leq m(m+1)$ we simply let $G$ be the output of the reduction algorithm. Since the graph does not change, the solution lifting algorithm simply outputs the solution it is given. In both cases, the reduced instance has size at most $m(m+1)$.
  
  It remains to verify that there is a $\ph$-approximation algorithm for \textsc{Clique Cover}. Given a graph $G$ of treewidth $\ell$, we simply output $\{\{v\} \mid v \in V(G)\}$. Clearly, this is a valid clique cover of $G$ of size $|V(G)|$. Observe that since $G$ has treewidth $\ell$, $G$ contains no cliques of size larger than $\ell + 1$, thus any clique in the optimal clique cover of $G$ covers at most $\ell + 1$ vertices. As such, the optimal solution contains at least $\frac{|V(G)|}{\ell+1}$ cliques, and thus $|S| \leq (\ell+1)\OPTp(G)$.
  
  \item[Feedback Vertex Set] Requirements~\ref{nice:union} and~\ref{nice:induced} are straightforward to verify. The problem has a $1$-approximate kernel with $4k^2$ vertices and therefore a PSAKS by Lemma~\ref{lem:fvs:approxkernel}, showing Requirement~\ref{nice:psaks}. It is also known that the \textsc{Feedback Vertex Set} problem has a $2$-approximation algorithm \cite{BECKER1996167}, showing Requirement~\ref{nice:approx}.

  \item[Edge Dominating Set] Requirement~\ref{nice:union} is again straightforward. For the second requirement, let $G$ be a graph and let $X \subseteq V(G)$. We start by showing that $\OPTp(G) \geq \OPTp(G-X)$. Let $S$ be an edge-dominating set in $G$. We obtain an edge-dominating set $S'$ for $G-X$ as follows. Initialize $S'$ as the set of edges with both endpoints in $V(G)\setminus X$, so $S' := \{e \in S \mid e \cap X = \emptyset\}$. For every edge $\{x,v\} \in S$ with $x \in X, v \notin X$, choose one arbitrary edge $\{u,v\} \in E(G-X)$ and add $\{u,v\}$ to $S'$. If no such edge exists, do nothing. Clearly, $|S'| \leq |S|$. Furthermore, we show that $S'$ is indeed an edge dominating set. Suppose for contradiction that $e = \{u,v\}$ is not dominated in $G-X$ by $S'$. Let $\{w,v\} \in S$ be the edge dominating $\{u,v\}$ in $G$. Then, since $\{w,v\} \notin S'$, we have $w \in X$. But then some edge with endpoint $v$ was added to $S'$, a contradiction. 
  
  We continue by showing that $\OPTp(G) \leq \OPTp(G-X) + |X|$ and that algorithm $\mathcal{A}$ exists. Let $S$ be a solution for $G-X$, then algorithm $\mathcal{A}$ will output $S$ together with one edge $\{x,v\} \in E(G)$ for all $x \in X$. In the case that $x \in X$ is isolated in $G$, no edge is added for this vertex. By this definition, the output has size at most $|S| + |X|$. Furthermore, any edge with vertices in $V(G-X)$ is dominated by $S$. Any edge with at least one endpoint in $X$ is dominated by the additional edges.
  
\textsc{Edge Dominating Set} has a kernel that outputs a graph $G'$ of size at most $4k^2 + 4k$ such that $G'$ is an induced subgraph of $G$ and any size-$k$ edge dominating set in $G'$ is also an edge dominating set in $G$ \cite{DBLP:conf/mfcs/Hagerup12}. We can see that this is a $1$-approximate kernel. Let the solution lifting algorithm simply output the solution for $G'$ as a solution for $G$. Since any solution of size at most $k$ in $G'$ is a solution in $G$, and obviously any solution in $G$ is a solution for $G'$, it is clear that $\modOPTp(G',k) = \modOPTp(G,k)$. As such, the approximation ratio is preserved by the solution lifting algorithm.
  
It is known that even the weighted version of \textsc{Edge Dominating Set} can be $2$-approximated \cite{FUJITO2002199}, such that the problem has a $\varphi$-approximation for $\varphi(s,\ell) = 2s$.\qedhere
 \end{description}
\end{proof}

As an  immediate consequence of  Lemma~\ref{lem:general-result}  and Theorem~\ref{thm:TK}, we obtain approximate Turing kernels for a large number of graph problems. These results are summarized in the corollary below, the size bounds are obtained by substituting the relevant bounds given by Lemma~\ref{lem:general-result} into Theorem~\ref{thm:TK}.

\begin{corollary}
 The following problems have a polynomial $(1+\varepsilon)$-approximate Turing kernel for all $0 < \varepsilon \leq 1$, of the given size (in number of vertices), when parameterized by treewidth~$\ell$.
 \begin{itemize}
  \item \textsc{Independent Set}, of size $\Oh( \frac{\ell^4}{\varepsilon^2} )$. 
  \item \textsc{Vertex-Disjoint $H$-packing} for connected graphs $H$, of size $\Oh((\frac{\ell}{\varepsilon})^{|V(H)|-1})$.
  \item \textsc{Vertex Cover}, of size $ \Oh(\frac{\ell}{\varepsilon})$.
  \item \textsc{Clique Cover}, of size $\Oh( \frac{\ell^4}{\varepsilon^2} )$.
  \item \textsc{Feedback Vertex Set}, of size $\Oh(( \frac{\ell}{\varepsilon})^2)$.
  \item \textsc{Edge Dominating Set}, of size $\Oh((\frac{\ell}{\varepsilon})^2)$.
 \end{itemize}
\end{corollary}

We observe that the bounds for \textsc{Independent Set} and \textsc{Clique Cover} can be improved to $\Oh(\frac{\ell^2}{\varepsilon})$ by a more careful analysis. Instead of using that the problem is friendly and applying Lemma~\ref{lem:find_t_minimization}, one may simply find $t$ such that the number of vertices in $G[V_t\setminus X_t]$ is between $\frac{(\ell+1)^2}{\delta}$ and $\frac{10(\ell+1)^2}{\delta}$, and use that an optimal solution has size at least $|V(G)|/(\ell+1)$ for graphs of treewidth $\ell$. There is no need to apply a kernelization in this case.

\section{Conclusion}
In this paper we have provided  approximate Turing kernels for various graph problems when parameterized by treewidth. Furthermore, we give a general result that can be used to obtain approximate Turing kernels for all \friendly graph problems parameterized by treewidth.

While the notion of being \friendly captures many known graph problems, some interesting problems do not fit this definition. In particular, it is not clear whether the \textsc{Dominating Set} problem has a polynomial-size  constant-factor approximate Turing kernel when parameterized by treewidth. We leave this as an open problem. %The problem does not fit our overall framework, as the idea of finding a subtree of the tree decomposition in which the solution is large cannot be simply applied: removing vertices may increase the size of an optimal dominating set arbitrarily much. We leave this as an interesting open problem. \todo{First of all, do we want to mention this at all  and second of all, rewrite?}

\bibliography{references}

\clearpage
\appendix
\section{Results on existing kernels}
\label{sec:existing-kernels}
In this section we will show for various classical kernels that they are in fact $1$-approximate kernels, which will be necessary for our approach. While this statement seems to be true for many kernels we know, it is not immediate from the definitions.

\begin{lemma}
 \label{lem:fvs:approxkernel}
 \textsc{Feedback Vertex Set} parameterized by solution size has a $1$-approximate kernel with $4k^2$ vertices.
\end{lemma}
\begin{proof}
 While it is well-known that \textsc{Feedback Vertex Set} has a small kernel parameterized by the solution size~\cite{10.1145/1721837.1721848}, we need to show that this is a $1$-approximate kernel to be able to use it. Let \P denote the \textsc{Feedback Vertex Set} problem. We show that all reduction rules used in \cite{10.1145/1721837.1721848} are $1$-safe, meaning  that for each reduction rule there is a lifting algorithm that  given $G,k$ and the result $G',k'$ of applying the reduction rule, and a solution $S$ to $G'$, outputs a solution $S$ to $G$ such that $\frac{\modP(G,k, S)}{\modOPTp(G,k)} \leq \frac{\modP(G',k', S')}{\modOPTp(G,k')}$ (for maximization problems we would require $\geq$).
  We check the reduction rules \cite[Page 32:5]{10.1145/1721837.1721848} one by one.
  \begin{description}
   \item[Rule 0] There is no solution in $G$ of size at most $k$, and there is no solution in $G'$ of size at most $k'$. Irrespective of the given solution, the solution lifting algorithm may output $S:= V(G)$, which has value $k+1$, which is optimal.
   \item[Rule 1] Let $S'$ be a solution in $G'$. Let $S:= S' \cup \{v\}$. Then $\modP(G,k,S) \leq \modP(G',k',S') + 1$. Furthermore, it is easily observed that $\modOPTp(G,k) = \modOPTp(G',k') + 1$, such that
   \[\frac{\modP(G,k, S)}{\modOPTp(G,k)} =  \frac{\modP(G',k', S')+1}{\modOPTp(G',k')+1} \leq \frac{\modP(G',k', S')}{\modOPTp(G,k')},\]
   where the last inequality follows from the fact that $\modP(G',k', S') \geq \modOPTp(G,k')$ by definition.
   \item[Rule 2] Observe that in this case, any feedback vertex set in $G$ is a feedback vertex set in $G'$, and vice versa.
   \item[Rule 3] Let $S'$ be a feedback vertex set in $G'$, then $S'$ is a feedback vertex set in $G$ of the same size. Furthermore it can be shown that $\modOPTp(G',k') = \modOPTp(G,k)$, refer to the correctness proof of the reduction rule for a proof.
   \item[Rule 4] In this case $G$ has no vertex cover of size at most $k$. Let the solution lifting algorithm output $V(G)$ (observe $\modP(G,k,V(G)) = k+1$). This is clearly a feedback vertex set, and $\frac{\modP(G,k,S)}{\modOPTp(G,k)} =1$, which is best-possible.
   \item[Rule 5] Let $S'$ be a solution in $G'$, let $S := S' \cup \{x\}$. Clearly $S$ is a feedback vertex set in $G$, as $G' = G - \{x\}$. We show that it has the right size. From the correctness of the reduction rule, we get $\OPTp(G) = \OPTp(G') + 1$. It follows $\modOPTp(G,k) = \modOPTp(G',k'=k-1)+1$. Furthermore, $\modP(G',k',S') +1= \modP(G,k,S)$, such that 
    \[\frac{\modP(G,k, S)}{\modOPTp(G,k)} = \frac{\modP(G',k', S')+1}{\modOPTp(G',k')+1} \leq \frac{\modP(G',k', S')}{\modOPTp(G,k')}.\]
    \item[Rule 6] It is shown in \cite[Theorem 3.2]{10.1145/1721837.1721848} that $\OPTp(G) = \OPTp(G')$, implying (by $k' = k$) that $\modOPTp(G',k') = \modOPTp(G,k)$. Furthermore, it is shown that any feedback vertex set in $G'$ is a feedback vertex set in $G$, concluding the proof.\qedhere
  \end{description}
\end{proof}

\begin{lemma}\label{lem:kernel:EDTP}
 \textsc{$k$-Edge-Disjoint Triangle Packing} parameterized by solution size has a $1$-approximate kernel with $4k$ vertices. 
\end{lemma}
\begin{proof}
 It is known that the problem has a kernel with $4k$ vertices~\cite{10.1007/978-3-540-28639-4_12}. We note that smaller kernels are known (cf. \cite{YANG2014344,LIN201920}), but we will use the size-$4k$ kernel as it will be easier to verify that this kernel is $1$-approximate 

We again show that the used reduction rules are $1$-safe when applying the reduction rules with parameter $\hat{k} := k + 1$ (note that we study the $1$-approximate kernel with respect to $k$). We verify all reduction rules given in \cite{10.1007/978-3-540-28639-4_12} below.
 \begin{description}
  \item[Rule 1] Let $(G',\hat{k}-1)$ be the instance obtained by applying reduction rule 1 on $(G,\hat{k})$, for vertices $u,v,w$. Given any solution $S'$ for $G'$, we may obtain solution $S := S' \cup \{\{u,v,w
 \}\}$. Clearly, this is a valid solution and it has value $|S'| + 1$. This will be our solution lifting algorithm. It is easy to observe the $\OPTetp(G) = \OPTetp(G') + 1$ and thus
  \[\frac{\modetp(G,k, S)}{\modOPTetp(G,k)} =  \frac{\modetp(G',k-1,S') + 1}{\modOPTetp(G',k-1) + 1} \geq \frac{\modetp(G',k-1,S')}{\modOPTetp(G',k-1)}.\]
  \item[Rules 2 and 3] It is easy to see that for both these reduction rules, $S$ is a solution to edge-disjoint triangle packing in $G$ if and only if it is a solution in $G'$. As such, the solution lifting algorithm simply outputs the given solution and the approximation factor is maintained.
  \item[Rule 4] Let the solution lifting algorithm output $S = S' \cup \{\{h_1,h_2,f(h)\}\mid h = \{h_1,h_2\} \in H\}$. Observe that $S$  is indeed an edge-disjoint triangle packing, as $\{h_1,h_2,f(\{h_1,h_2\})\}$ is a triangle in $G$ by definition, the triangles in $\{\{h_1,h_2,f(h)\}\mid h = \{h_1,h_2\} \in H\}$ are edge-disjoint by the fact that $f$ is injective and furthermore disjoint from $S'$ as $S'$ does not use vertices from $C\cup V(H)$. So, $S$ is an edge-disjoint triangle packing in $G$ of size at least $|S'| + |H|$. This immediately shows that $\modOPTetp(G',k - |H|) + |H| \leq \modOPTetp(G,k)$. To show that $\modOPTetp(G',k - |H|) + |H| = \modOPTetp(G,k)$, suppose we are given a solution $S$ in $G$. Obtain $S'$ by removing all triangles from $S$ that contain at least one vertex from $V(H) \cup C$. Clearly, this results in an edge-disjoint triangle packing in $G'$. It remains to show that $|S'| \geq |S|-|H|$. Observe however that there are at most $|H|$ triangles containing (one or more) vertices from $V(H)\cup C$, since any triangle containing a vertex in $C$ must use at least one edge in $H$ since $C$ is an independent set and there are no edges between $C$ and $X$. To conclude,
   \[\frac{\modetp(G,k, S)}{\modOPTetp(G,k)} =  \frac{\modetp(G',k-|H|,S') + |H|}{\modOPTetp(G',k-|H|) + |H|} \geq \frac{\modetp(G',k-|H|,S')}{\modOPTetp(G',k-|H|)}.\]
  
  \item[Kernelization Lemma] The lemma states: ``If $G$ is reduced under reduction rules $1$ to $4$, and $V(G) > 4\hat{k}$, then $G$ is a yes-instance for \textsc{$\hat{k}$-Edge-Disjoint Triangle Packing}.'' \cite{10.1007/978-3-540-28639-4_12}. In other words, if $G$ is reduced under the reduction rules above and $V(G) > 4\hat{k}$ then the ($1$-approximate) kernelization algorithm can output a trivial yes-instance. The solution lifting algorithm will then output a size-$\hat{k}$ triangle packing in $G$ (regardless of the solution given for the kernelized instance), which has value $\hat{k}  = k+1 \geq \modOPTetp(G,k)$. It remains to show how to find such a triangle packing of size $\hat{k}$. Start from any packing $S$. As long as there is a triangle that can be added to $S$ as none of its edges are covered by $S$, do so. Furthermore, as long as there is a triangle $\{u,v,w\} \in P$ and $x,y$ not contained in any triangle such that $\{x,u,v\}$ and $\{y,v,w\}$ are triangles, let $S \gets S \cup \{\{x,u,v\},\{y,v,w\}\}\setminus \{\{u,v,w\}\}$. Continue until neither of these two rules apply, note that both steps increase $|S|$, such that this procedure will halt after at most $|E(G)|/3$ steps. 
  
  We show $|S| \geq \hat{k}$,  we reuse some of the proof strategy given by~\cite{10.1007/978-3-540-28639-4_12}. Assume for contradiction that $|S| \leq \hat{k}$. Let $O$ be the set of vertices not contained in any triangle in $S$. Say a vertex $u$ \emph{spans} edge $\{v,w\}$ in $G$ if $\{u,v\} \in E(G)$ and $\{u,w\} \in E(G)$. Slightly changing the notation from \cite{10.1007/978-3-540-28639-4_12}, define $S_0$ to be the triangles in $S$ with no vertices in $O$ spanning any of their edges, $S_1$ to be the triangles where exactly one edge is spanned by a vertex in $O$, and $S_A$ the triangles for which all edges are spanned. Observe that $S = S_0 \cup S_1\cup S_A$. Observe that Claims 1-6 \cite{10.1007/978-3-540-28639-4_12} still hold for our (different) choice of $S$, in particular $O$ is an independent set and every vertex in $O$ spans at least one edge in a triangle in $S$. Let $O_A$ be the subset of vertices in $O$ who span one or more edges of a triangle in $S_A$, let $O_1$ be the set of vertices who span one or more edges in $S_1$. Then $|O| = |O_1 \cup O_A| \leq |O_1| + |O_A|$. Using the same strategy as in  \cite{10.1007/978-3-540-28639-4_12}, it follows $|O_1| \leq |S_1|$ and $|O_A| \leq |S_A|$. Thus, $|O| \leq |S|$. Therefore, $|V(G)| \leq |O| + 3\cdot|S| \leq 4 \cdot |S| \leq 4\hat{k}$, contradicting that this reduction rule could be applied. \qedhere
 \end{description}
\end{proof}

\begin{lemma}
 \label{lem:packing:approxkernel}
 \textsc{Vertex-Disjoint $H$-Packing} parameterized by solution size has a $1$-approximate kernel with $\Oh(k^{|V(H)|-1})$ vertices.
\end{lemma}
\begin{proof}
 We show this by showing that the kernel given in \cite{DBLP:conf/sofsem/Moser09} is in fact a $1$-approximate kernel when applied with parameter $\hat{k} := k + 1$. The kernel consists of three reduction rules. We again verify that they are $1$-safe.
  \begin{description}
   \item[Rule 4] Trivially, any $H$-packing in $G'$ is also an $H$-packing in $G$, and vice versa.
   \item[Rule 5] Clearly, any $H$-packing in $G'$ is an $H$-packing in $G$. Furthermore, it is shown in \cite[Lemma 6]{DBLP:conf/sofsem/Moser09} that, given a packing $S$ in $G$, there is a packing in $G'$ with size at least $|S|$. As such, $\OPTp(G) = \OPTp(G')$ and the solution lifting algorithm may simply return $S'$.
   \item[Rule 6] Any solution to $G'$ is a solution for $G$. It remains to show that $\modOPTp(G',k') = \modOPTp(G,k)$, to conclude this case. In the proof of correctness of this reduction rule it is indeed shown that, given a packing in $G'$, we can always replace any copies of $H$ that contain a removed vertex, by a copy that does not, obtaining a packing of the same size in $G$. 
   \item[{\cite[Lem 10]{DBLP:conf/sofsem/Moser09}}] The kernel uses an additional lemma showing that if certain size-lower bounds are met, then $G$ has an $H$-packing of size $\hat{k}$. If this case is encountered, the solution lifting algorithm may simply output such a packing, which will have value $\hat{k} = k +1$, which is optimal. \qedhere
  \end{description}
\end{proof}

\end{document}